\documentclass[journal]{IEEEtran}

\usepackage{xcolor}
\usepackage{soul}
\sethlcolor{yellow}

\usepackage{amsmath,amssymb,amsfonts}
\usepackage{mathtools}
\usepackage{amsthm}
\usepackage{bm}

\usepackage{algorithm}
\usepackage{algorithmic}

\usepackage{graphicx}
\usepackage[caption=false,font=normalsize,labelfont=sf,textfont=sf]{subfig}
\usepackage{stfloats}
\usepackage{array}
\usepackage{tabularx}
\usepackage{multirow}
\usepackage{makecell}
\usepackage[table]{xcolor}

\usepackage{enumitem}
\usepackage{textcomp}
\usepackage{url}
\usepackage{verbatim}
\usepackage{cite}
\usepackage{hyperref}
\usepackage{xurl}

\newcounter{appx}
\renewcommand{\theappx}{\Alph{appx}}

\definecolor{HeaderBlue}{HTML}{0E5A73}
\definecolor{SubHeaderGray}{HTML}{C9CED4}
\definecolor{BodyGray}{HTML}{E9ECEF}

\newcolumntype{Y}{>{\centering\arraybackslash}X}

\newtheorem{thm}{Theorem}
\newtheorem{defi}{Definition}
\newtheorem{lemma}{Lemma}
\newtheorem{prop}{Proposition}
\newtheorem{rem}{Remark}

\newcommand{\Zblk}[2]{%
\begin{bmatrix}
Z_{dd}^{#1#2} & Z_{dq}^{#1#2}\\
Z_{qd}^{#1#2} & Z_{qq}^{#1#2}
\end{bmatrix}}

\title{Decentralized Stability of IBR-dominated Power Grids Using Block Diagonal Dominance}

\author{Youhong~Chen,~\IEEEmembership{Member,~IEEE},
Xiaoyu~Tan,~\IEEEmembership{Student~Member,~IEEE},
Muhammad~Sharjeel~Javaid,~\IEEEmembership{Member,~IEEE},
David~Angeli,~\IEEEmembership{Fellow,~IEEE},
and~Balarko~Chaudhuri,~\IEEEmembership{Fellow,~IEEE}
\thanks{This work was supported by the EPSRC under Grant EP/Y025946/1. Corresponding author: Xiaoyu Tan (e-mail: \href{mailto:xiaoyu.tan22@imperial.ac.uk}{xiaoyu.tan22@imperial.ac.uk}).}
\thanks{The authors are with University of Bath and Imperial College London, U.K.}}

\IEEEpubid{}
\begin{document}
\maketitle
\bstctlcite{TPWRScontrol}

% \begin{abstract}
% This paper presents a decentralized stability criterion for inverter-based resource (IBR)–dominated power systems based on a diagonal dominance condition applied to open-loop transfer functions. In contrast to conventional closed-loop stability assessments that rely on eigenvalue analysis, the proposed framework evaluates system stability by integrating the dynamics of individual power devices and the network in an open-loop and fully decentralized manner. Compared with existing decentralized stability conditions, including small-gain, small-phase, and passivity-based criteria, the proposed method exhibits reduced conservatism. Moreover, the proposed approach is readily generalizable and well suited for practical connection compliance assessment in large-scale multi-IBR power systems. The effectiveness of the method is demonstrated using a modified IEEE 39-bus test system.
% \end{abstract}

\begin{abstract}
    The growing penetration of inverter-based resources (IBRs) necessitates stability assessment methods that are scalable, decentralized, and model-agnostic. This paper develops a block diagonal dominance (BDD) criterion for decentralized small-signal stability of IBR-dominated power grids. The proposed approach forms the basis for an enhanced IBR connection compliance condition from a small-signal stability perspective that can be evaluated locally for IBRs to be connected to the grid. The proposed approach is shown to be much less conservative than strict diagonal dominance (SDD). Beyond mere stability, we ensure a minimum decay rate or maximum settling time for IBR-induced oscillation. Crucially, these are achieved without imposing restrictive assumptions on network or IBR models. The framework therefore, offers a practical and theoretically grounded basis for decentralized stability certificate of IBR-dominated power grids.
\end{abstract}


\section{Introduction}
\IEEEPARstart{D}{ecarbonization} of the electricity sector worldwide is driving unprecedently high penetration of inverter-based resources (IBRs) including wind and solar photovoltaic generation, grid-scale battery storage and HVDC links in power grids. While these IBRs are essential for clean energy transition, their widespread deployment fundamentally alters grid dynamics and introduces new stability challenges \cite{Hatziargyriou2021Definition}. In particular, IBR-dominated grids are increasingly susceptible to poorly damped sub-synchronous oscillations (SSOs) which are being experienced in different parts of the world \cite{Cheng2023Events,PEMag24_Modi_Real_world_SSO_events}.

The major difficulty in assessing the stability of IBR-dominated systems lies not only in their scale and complexity, but also in the lack of transparency of vendor-specific IBR models \cite{esig_oscillations_2024}. Unlike synchronous machines, whose dynamics are well understood and standardized, IBR control implementations are often proprietary which makes systematic, system-wide stability analysis using traditional model-based approaches difficult. As a result, grid operators face growing uncertainty when evaluating the stability impact of new IBR connections even though it is critical for them to foresee the risk of IBR-induced SSO during the connection process itself.

At present, IBR connection compliance (interconnection studies) typically relies on detailed models of the connecting IBR plant combined with simplified representations of the grid, often reduced to a Thevenin equivalent \cite{NESO2025Guidance}. Stability and performance (e.g., a minimum decay rate) is then assessed using time-domain simulations (step response) or frequency-domain techniques such as eigenvalue or Nyquist analysis \cite{NESO2025Guidance}. A Thevenin equivalent at fundamental frequency is not expected to capture interaction between the connecting IBR and the grid across the whole sub-synchronous frequency range and with existing IBRs in the grid. A more robust alternative is to use the estimated frequency response (or corresponding vector-fitted transfer function) of the grid (including existing IBRs) to capture potential IBR-grid interaction across the frequency range of interest \cite{aemo2025frequency}. While suitable for a particular upcoming IBR, this approach fails to capture dynamic interactions with other upcoming IBRs. Consequently, stability guarantees obtained at the time of a particular IBR connection may no longer hold as additional IBRs connect elsewhere in the grid. Ensuring stability and a specified performance (say, minimum decay rate of SSO) as the grid evolves with new IBR connections necessitates a decentralized stability framework, where each IBR can be assessed for compliance locally, yet overall system stability and performance is preserved.
% regardless of the number of IBR connections over a given planning horizon.

In this paper, we present such a framework using diagonal dominance criterion\cite{Horn1985,rosenbrock1974}. The basic theory is presented in Sections II and III but the idea is to derive decentralized small-signal stability conditions that can be checked independently and locally at each IBR connection point using the estimated transfer functions of the connecting IBR and the grid with existing IBRs. When all connecting IBRs satisfy their prescribed conditions, overall system stability and performance metrics such as decay rate or settling time is ensured. In power system literature, diagonal dominance has been used in conjunction with Nyquist array theory to perform system-level stability assessment and controller design \cite{DD_PS_Paper} for a single IBR only. This work, however, proposes the use of diagonal dominance as a decentralized, locally verifiable certificate at individual IBR connection points for a multi-IBR system, rather than as a tool to perform model-based Nyquist array analysis and control design for a single IBR.

To illustrate the application, let’s consider a set of $N$ candidate locations (or buses) where new IBRs are expected to connect over a given planning horizon. The estimated transfer function of the grid, including existing IBRs, can be obtained at these locations via frequency scanning \cite{GIST2023Shah} and represented as a $2N\times2N$ transfer function in $dq$-domain using vector fitting \cite{vectorfitting}. As each IBR is a $2\times 2$ transfer function block, the block diagonal dominance (BDD) criterion developed in this paper enables each upcoming IBR to be assessed locally against the relevant block row of the grid model. Crucially, although the assessment is local, the resulting stability and performance guarantee is global, provided all IBRs satisfy their prescribed conditions derived from BDD criterion. This is radically different from the current stability check at each IBR connection point without consideration of prospective other upcoming IBR connections.

A well-known challenge with decentralized stability approach is conservativeness. %\cite{decentralized_conservativeness}.
As the number of candidate IBR connection points increases especially, across electrically distant parts of the grid, decentralized stability criteria like strict diagonal dominance (SDD) may impose increasingly restrictive compliance requirements on connecting IBRs. Reducing this conservativeness is therefore crucial for practical adoption. Otherwise, connecting IBR plants would face overly strict compliance condition which might not be easy to meet making it difficult for grid operators to clear the IBR connection queue. This represents a trade-off between ensuring stability certificate and avoiding unnecessary restrictions on IBR integration. The proposed BDD-based decentralized framework provides grid operators with a practical means to balance this trade-off.

% This is a trade-off between stability certificate while letting the IBRs connect without unnecessary restriction, which the grid operators will have to properly calibrate. The decentralized stability framework based on BDD in this paper will enable grid operators to do so. 

% Decentralized methods, such as small-gain and small-phase (SGP) condition \cite{10477540}, have been proposed for IBR-dominated grids, but often rely on restrictive assumptions or exhibit significant conservativeness in realistic networks. This framework was further extended in \cite{huang2025DW} through the development of a Davis–Wielandt shell–based complex matrix geometric analysis method, which combines gain and phase conditions in a three-dimensional space to reduce conservativeness. Other decentralized stability conditions have also been proposed, including dynamic loop-shifting techniques in \cite{haberle2025} and dissipativity-based criteria in \cite{8815214,gorbunov2026}, which reshape system dynamics to enforce passivity properties and to achieve small-signal stability in a plug-and-play manner for IBR integration. However, most of these methods rely on restrictive assumptions regarding network and IBR modeling, which limits their applicability in practical system settings. 
% In fact, one could envisage a unified decentralized stability condition where it is sufficient to satisfy any one condition at each frequency to reduce overall conservativeness.
Decentralized methods, such as small-gain and small-phase (SGP) condition \cite{10477540}, have been proposed for IBR-dominated grids, and this framework was further extended in \cite{huang2025DW} through the development of a Davis–Wielandt shell–based complex matrix geometric analysis method, which combines gain and phase conditions in a three-dimensional space to reduce conservativeness. Other decentralized stability conditions have also been proposed, including dynamic loop-shifting techniques in \cite{haberle2025} and dissipativity-based criteria in \cite{8815214,gorbunov2026}, which reshape system dynamics to enforce passivity properties and to achieve small-signal stability in a plug-and-play manner for IBR integration. However, much of the literature relies on simplifying assumptions for analytical tractability, both in network representations such as identical $X/R$ ratios and zero power flow network models, and in IBR modeling such as simplified dynamics that omit inner loop dynamics and assume known control settings, thereby limiting applicability in practical system settings. 

In this paper, we demonstrate that a BDD-based condition is much less conservative than an SDD-based condition in terms of admissible IBR dynamics and electric proximity of IBR connection points. We not only ensure decentralized stability but also achieve a minimum decay rate of SSO. %Importantly, these benefits are achieved without imposing restrictive assumptions on network or IBR models.

This paper makes the following original contributions:
\begin{enumerate}[leftmargin=*]
\item Extends the strict diagonal dominance (SDD) criterion to the block diagonal dominance (BDD) for decentralized small-signal stability of IBR-dominated power grids.
\item Proves and demonstrates that BDD criterion is much less conservative than SDD criterion.
\item Establishes an enhanced IBR connection compliance condition based on BDD to reduce the risk of small-signal instability, mainly IBR-induced sub-synchronous oscillations.
\item Ensures a minimum decay rate (or maximum settling time) of oscillations in addition to mere stability certificate
\item Achieves the above without any restrictive assumptions on the network or IBRs, unlike competing approaches in the literature.
\end{enumerate}

\section{Background: MIMO Linear Time-Invariant Systems}
\label{sec:power system modelling}
This section provides background notations and results on the stability of MIMO linear time-invariant (LTI) systems in feedback.
We consider continuous-time LTI MIMO systems of the following form:
% \begin{equation}
%     \label{LTI} 
%     \begin{array}{rcl} \dot{x} &=& A x + B u \\
%     y&=& C x
%     \end{array}
\begin{equation}
\label{LTI}
\dot{x} = A x + B u,\qquad y = C x
\end{equation}
with state variable $x \in \mathbb{R}^n$, input $u \in \mathbb{R}^m$ and output $y \in \mathbb{R}^m$. As standard, $A$, $B$ and $C$ denote real matrices of suitable dimension.
We associate to a system as in (\ref{LTI}), its input-output transfer function:
\begin{equation*}
    H(s) := C (sI_n - A)^{-1} B
\end{equation*}
which is a proper rational matrix function $H: \mathbb{C} \rightarrow \mathbb{C}^{m \times m}$ in the unknown $s \in \mathbb{C}$. Asymptotic (internal) stability of system (\ref{LTI}) is characterized in terms of the spectrum of $A$, denoted $\textrm{sp} (A)$, i.e.
\begin{equation*}
\textrm{sp} (A) := \{ s \in \mathbb{C}: \textrm{det} (s I_n -A ) = 0 \}.
\end{equation*} 
In particular, asymptotic stability holds iff 
\[ \textrm{sp} (A) \subseteq \mathbb{C}_{<0} :=  \{ s: \textrm{Re} (s) <0 \}. \]
In addition, we denote by $\textrm{poles}(H)$ the set of poles of the (matrix) transfer function $H$. For a matrix function these correspond to the poles of its individual entries, i.e. $\textrm{poles}(H) = \bigcup_{i,j={1}}^{m} \textrm{poles} ( H_{ij}  )$.
We denote by $\mathcal{RH}_{\infty}$ the set of BIBO stable transfer functions, i.e. characterized by the inclusion: $\textrm{poles}(H) \subset \mathbb{C}_{<0}$.
It is well known that $\textrm{poles}(H) \subseteq \textrm{sp}(A)$,
so that \emph{asymptotic} stability implies \emph{BIBO} stability. The converse implication needs not hold, in particular if
  $\textrm{sp}(A) \backslash \textrm{poles} (H) \neq \emptyset$ its elements corresponds to eigenvalues of $A$ which are either unreachable or unobservable (or both). When these eigenvalues do not belong to $\mathbb{C}_{<0}$, then the system is BIBO stable but not asymptotically stable. Notice that such eigenvalues are unaffected by output feedback, and therefore are not the focus of our analysis. This allows, for instance, to take into account open (and closed-loop) models which may exhibit unobservable or unreachable eigenvalues in $0$, i.e. $0 \in \textrm{sp}(A) \backslash \textrm{poles} (H) $. 
  For more stringent performance qualifications, it is useful to consider 
  proper subsets of $\mathbb{C}_{<0}$. In the following, we denote by
  $\mathbb{C}_{< \alpha}$ the set
  \[ \mathbb{C}_{<\alpha} := \{ s \in \mathbb{C}: \textrm{Re}(s) < \alpha \}, \]
  where $\alpha<0$ is a suitable real parameter chosen to ensure transient modes with adequate exponential decay rate. \\
Next, we introduce the closed-loop system induced by the negative unity feedback $u = v - y$, where $v \in \mathbb{R}^m$ is an auxiliary exogenous input that enters the system through the same channel as $u$. The state-space equation reads:
\begin{equation}
    \label{LTIcl} 
 \dot{x} = (A-BC) x + B v, \quad  y= C x
 \end{equation}
and its transfer function can be expressed as:
$H_{cl}(s) =
C (s I_n - (A-BC) )^{-1} B$. Alternatively, using the open loop transfer function we have:
\begin{equation*}
 \label{iotf}
 H_{cl} (s) = H(s) ( I_m + H(s) )^{-1} = (I_m +H(s) )^{-1} H(s).
\end{equation*}
We are interested in the location of the poles of the closed-loop transfer function $H_{cl}(s)$. 
To this end denote by $\textrm{zeros} ( R )$ the zeros of a scalar rational function $R: \mathbb{C} \rightarrow \mathbb{C}$. 
The following inclusions are crucial to relate properties of the open-loop system to those of the closed-loop system (see Thm. IV.12, \cite{desoer}):
\begin{equation*}
 \label{cl_loop}
 \begin{array}{rcl}
  \textrm{zeros} ( \textrm{det} ( I_m + H ) )  &\subseteq &\textrm{poles}(H_{cl}  ) \\ \textrm{poles}(H_{cl}  )  \subseteq \textrm{zeros} ( \textrm{det} ( I_m + H ) ) &\cup& \textrm{poles} (H).
  \end{array}
\end{equation*}
A corollary of the above is the following result:
\begin{thm}
\label{necsuf}
Consider an LTI system as in (\ref{LTI}) and the corresponding closed-loop system (\ref{LTIcl}). Assume that 
$\textrm{poles} ( H ) \subseteq \mathbb{C}_{< \alpha}$. Then,
$\textrm{poles} (H_{cl}) \subseteq \mathbb{C}_{< \alpha}$ iff 
$\textrm{zeros} (\textrm{det}(I_m + H) ) \subseteq \mathbb{C}_{< \alpha}$.
\end{thm}
An alternative approach to direct computation of the zeros, is to exploit continuity of $\textrm{zeros} ( \textrm{det} (I_m + k H ) )$ with respect to the homotopy parameter 
$k \in [0,1]$. 
Consider $\partial \mathbb{C}_{< \alpha}$ to be the boundary
of $\mathbb{C}_{< \alpha}$ i.e.
\[ \partial \mathbb{C}_{< \alpha} = \{ s: \exists \omega \in \mathbb{R}: s = \alpha + j \omega \}. \]
Theorem \ref{necsuf} can be adapted by considering the following sufficient condition:
\begin{thm}
\label{sufonly}
Consider an LTI system as in (\ref{LTI}) and the corresponding closed-loop system (\ref{LTIcl}). Assume that 
$\textrm{poles} ( H ) \subseteq \mathbb{C}_{< \alpha}$. Then,
$\textrm{poles} (H_{cl}) \subseteq \mathbb{C}_{< \alpha}$ 
provided
\[   \textrm{det} ( I_m + k H(s) ) \neq 0 \; \quad \forall \, s
\in \partial \mathbb{C}_{<\alpha}, \; \forall \, k \in (0,1]. \]
\end{thm}

\section{Diagonal Dominance Criterion} \label{sec:DD_theory}

\subsection{Strict Diagonal Dominance (SDD)} \label{subsec:SDD_theory}
The condition in Theorem \ref{sufonly} is a non-singularity test of a complex-valued square matrix over a set of possible frequencies and values of the homotopy parameter $k$.
To further simplify and ``decentralize'' its use we replace non-singularity
by the notion of Strict Diagonal Dominance (see \cite{Horn1985}). This method, pioneered by H. Rosenbrock in \cite{rosenbrock1974}, was initially adopted to simplify MIMO systems analysis and afford graphical criteria for controller synthesis akin to the SISO Nyquist plot. It is, nevertheless, still of interest for large scale networks, due to its decentralized nature.  
\begin{defi}
A matrix $M =[m_{ij}] \in \mathbb{C}^{m \times m}$ is \emph{strictly diagonally dominant}(SDD) if
\[ | m_{ii} | > \sum_{j \neq i} |m_{ij}|, \quad \forall \, i = 1 ,\ldots, m. \]
\end{defi}
It is well known that $M$ being strictly diagonally dominant implies $M$ is non-singular, (see Theorem 6.1.10, \cite{Horn1985}).
In the light of this consideration, we propose the following sufficient condition for stability of closed-loop LTI systems.

% \vspace{-0.8\baselineskip}
\begin{thm}
\label{DDsuf_}
Consider an LTI system as in (\ref{LTI}) and the corresponding closed-loop system (\ref{LTIcl}). Assume that 
$\textrm{poles} ( H ) \subseteq \mathbb{C}_{< \alpha}$. Then,
$\textrm{poles} (H_{cl}) \subseteq \mathbb{C}_{< \alpha}$ 
provided
\[   I_m + k H(s) \textrm{ is SDD }  \;  \forall \, s
\in \partial \mathbb{C}_{<\alpha}, \; \forall \, k \in (0,1]. \]
\end{thm}
Since $I_m+k H$ is SDD iff $\frac{1}{k} I_m + H$ is SDD, we may avoid a condition that ranges over two real parameter $\omega$ and $k$, by directly looking at the worst possible value of $1/k \in [1, + \infty)$ in the numerical verification of the SDD condition. 

\begin{figure}[!t]
\centering
\includegraphics[width=0.95\linewidth]{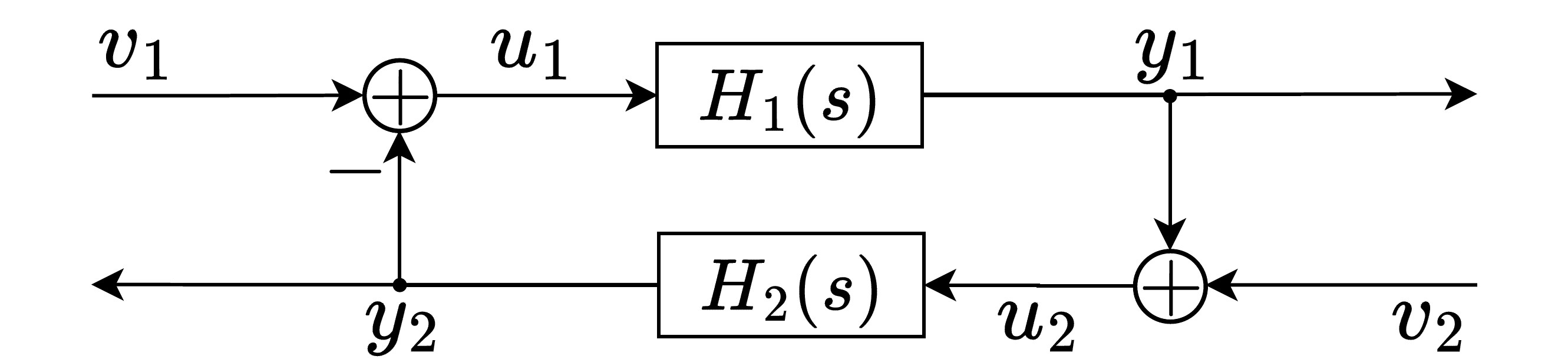}
\caption{Feedback of two subsystems.}
\label{fig:mimoloop}
\end{figure}

\begin{lemma}
The matrix $I_m + k H$ is SDD for all $k \in (0,1]$ if and only if:
\[ \min_{ \mu \in [1,+\infty) } | \mu + H_{ii} | > \sum_{j \neq i}  |H_{ij}|, \; \forall \, i = 1, \ldots , m. \]
\end{lemma}
The quantity on the left-hand side of the previous inequality can be explicitly computed:
\[ \min_{ \mu \in [1,+\infty) } | \mu + H_{ii} |
= \left \{ \begin{array}{cl} \textrm{Im} ( H_{ii} ) & \textrm{if } \textrm{Re} ( H_{ii} ) \leq -1 \\
|1 + H_{ii} | & \textrm{if } \textrm{Re} ( H_{ii} ) > -1
\end{array} \right .  . \]
so that the corresponding inequalities `only' need to be verified for all $s \in \partial \mathbb{C}_{< \alpha}$,
rather than searching the space also with respect to the homotopy parameter $k \in (0,1]$.
When feedback of two subsystems is considered, as in Fig. \ref{fig:mimoloop},
one is interested in the BIBO stability of all transfer functions, (from each input $v_i$, $i=1,2$ towards each output $y_i, u_i$, $i=1,2$). Notice that $u_i$ is now an endogenous signal and may be regarded as an output.
For this to be the case the so-called Gang of Four transfer functions need to be considered:
\[
  \left [   \begin{array}{c} u_1 \\ u_2
  \end{array} \right ] = \left [ \begin{array}{cc}
      (I+H_1 H_2)^{-1} &   - (I+H_1H_2) H_2 \\
     H_1 (I+H_2H_1)^{-1}  & (I+ H_2 H_1)^{-1}
  \end{array}\right ]  \, \left [ \begin{array}{c} v_1 \\ v_2
  
  \end{array}
  \right ]
\]
We propose the following sufficient condition to enforce this desirable stability property:
\begin{thm}
\label{DDsuf}
Consider the feedback interconnection in Fig. \ref{fig:mimoloop} of LTI system as in (\ref{LTI}) and the corresponding closed-loop system (\ref{LTIcl}). Assume that 
$\textrm{poles} ( H_i ) \subseteq \mathbb{C}_{< \alpha}$, for $i=1,2$. 
Then the poles of the Gang of Four transfer function are in $\mathbb{C}_{< \alpha}$ 
provided
\[   I + k H_1(s) H_2(s) \textrm{ is SDD }  \;  \forall \, s
\in \partial \mathbb{C}_{<\alpha}, \; \forall \, k \in (0,1]. \]

Alternatively, provided:
\[  \sum_{j \neq i} | [H_1(s) H_2(s)]_{ij} | < m_i(s) \, \quad \forall \, s \in \partial \mathbb{C}_{< \alpha}, \, \forall \, i \]
where
\end{thm}
\vspace{-2em}
\[ m_i=\left \{ \begin{array}{cl} \textrm{Im}( [H_1 H_2]_{ii} ) & \textrm{if } \textrm{Re} ( [H_1 H_2]_{ii} ) \leq -1 \\
| 1 + [H_1 H_2]_{ii} | & \textrm{if }\textrm{Re} ( [H_1 H_2]_{ii} ) > -1 \end{array} \right . . \]

\subsection{Block Diagonal Dominance (BDD)} \label{subsec:BDD_theory}
% When variables, as in a representation of the $dq$ power grid, can be naturally paired according to the location of each bus, relevant matrices can then be interpreted as block-partitioned and involving $2 \times 2$ blocks.

As shown in Fig.~\ref{fig:IBRs_RoG} and is discussed in Section \ref{sec:app DD in power system}, when the power system is modeled in the $dq$ reference frame, system variables can be naturally paired according to their bus locations. Consequently, the associated matrices can be interpreted in block-partitioned form, with each block having a dimension $2 \times 2$. This allows for tighter decentralized tests of non-singularity.
 Consider a block matrix $M$, as shown below:
\begin{equation}
\label{bmat}
M = \left [ \begin{array}{ccccc}  M_{1,1} & M_{1,2} & \ldots & & M_{1, Q} \\ M_{2,1} & M_{2,2} & \ldots & & M_{2,Q}  \\ \vdots & & \ddots&  & \vdots \\ M_{Q,1} & M_{Q,2} & \ldots & & M_{Q,Q}   \end{array} \right ],
\end{equation}
where $Q$ is the number of blocks considered. We aim to introduce an extension of the classic notion of diagonal dominance for block partitioned matrices along the lines of \cite{varga,vargabook}.
For any $i \in \{1, \ldots, Q \}$ we define $M_{i,-i}$ the following matrix:
\[    M_{i,-i} = [ M_{i,1}, M_{i,2}, \ldots, M_{i,i-1}, M_{i,i+1}, \ldots, M_{i, Q} ]. \]
We are now ready to state a sufficient condition for non-singularity of $M$. 
\begin{thm}
    \label{blockDD}
 Let $M$ be a square complex valued block matrix as in (\ref{bmat}) and such that for all $i \in \{1,\ldots,Q \}$
it holds:
\begin{equation}
\label{relaxedDD}
1 > \| M_{i,i}^{-1} M_{i,-i} \|_{\infty}.
\end{equation}
Then, $M$ is non-singular. \    
\end{thm}
For the sake of completeness, the proof of Theorem \ref{blockDD} is included in Appendix \ref{blockthmproof}.
Notice that, in the special case of $1 \times 1$ blocks, the condition boils down to the Strict Diagonal Dominance (SDD). However,  for the case of $2 \times 2$ blocks, it can be shown to be a milder requirement, as clarified below:
\begin{prop}
\label{blockbetter}
Assume that the strict diagonal dominance condition holds for $i=1,2$, viz.
\begin{equation*}
    \label{hp1}
\begin{array}{rcl}   |m_{11}|  &>& \sum_{j \neq 1} |m_{1j}| \\
 |m_{22}| &>& \sum_{j \neq 2} |m_{2j}|.
 \end{array}
 \end{equation*}
Then, the block diagonal dominance condition (\ref{relaxedDD}) holds, viz.
\[     1 >    \left \| \left [ \begin{array}{cc}  m_{11} & m_{12} \\ m_{21} & m_{22} \end{array} \right ]^{-1} \, \left [ \begin{array}{cccc} m_{13} & m_{14} & \ldots& m_{1n} \\ m_{23} & m_{24} & \ldots & m_{2n} \end{array} \right ] \right \|_{\infty}.  \]
\end{prop}

\emph{Proof.} See Appendix \ref{proofnonconservative}.
It is worth pointing out that the converse might not hold.

\begin{rem}
Although (\ref{relaxedDD}) is inspired by the theory of diagonal dominance for block-partitioned matrices, to the best of our knowledge it is a novel condition, milder than the classical one:
\[ \sum_{j \neq i} \| M_{ii}^{-1} M_{ij} \|_{\infty } < 1.  \]
The main reason for (\ref{relaxedDD}) not being picked up by the classical literature appears to be the aim of proposing a condition whose structure closely mimics that of SDD.
On this note, an even stronger requirement often reported in the literature is the following condition:
\[ \sum_{j \neq i} \| M_{ij} \|_{\infty } < \| M_{ii}^{-1} \|^{-1}.  \]
\end{rem}

\section{Decentralized Stability Conditions for Power Systems}\label{sec:app DD in power system}
%%%%%%%%%%%%%%%%%%%%%%%%%%%%%%%%%%%%%%%

%To capture inter-IBR coupling, define the off-diagonal block row
%\begin{equation}
%L_{i,-i}(s_{\alpha})
%:=
%\big[
%L_{i1}(s_{\alpha}) \ \cdots \ L_{i,i-1}(s_{\alpha}) \ 
%L_{i,i+1}(s_{\alpha}) \ \cdots \ L_{iN}(s_{\alpha})
%\big].
%\label{eq:Li_minus_i}
%\end{equation}

%%%%%%%%%%%%%%%%%%%%%%%%%%%%%%%%%%%%%%%%%

% In this section, we show how the block diagonal dominance (BDD) condition presented in Section~\ref{subsec:BDD_theory} can be applied in a decentralized way to ensure stability (and a minimum decay rate) of a power system during connection compliance of IBR plants, henceforth referred to as `IBR'.
In this section, we show how the block diagonal dominance (BDD) criterion presented in Section~\ref{subsec:BDD_theory} can be specialized to the $dq$-domain power setting to yield decentralized conditions for the connection compliance assessment of IBR plants, henceforth referred to as 'IBR', that ensure stability (and a minimum decay rate).

% \subsection{Decentralized Conditions Formulation in Power Systems Context}

Consider a power system with $N$ ($N>1$) upcoming IBRs to be connected at $N$ buses with the rest of the grid (RoG) as shown in Fig.~\ref{fig:IBRs_RoG}(a). These $N$ buses are the respective points of interconnection (PoI) of the IBRs. The RoG comprises existing IBRs, synchronous machines, loads, and the network. The overall system is partitioned into two subsystems - the upcoming IBRs and the RoG. Each IBR is represented as a $(2\times 2)$ admittance matrix in $dq$ domain. The subsystem of $N$ IBRs is thus a $(2N \times 2N)$ block-diagonal matrix given by:
\begin{equation} \label{eq:Y_block}
\mathbf{Y_I}(s) =
\begin{bmatrix}
\left[
\begin{matrix}
Y_{dd}^1 & Y_{dq}^1 \\
Y_{qd}^1 & Y_{qq}^1
\end{matrix}
\right] & 0 & \cdots & 0 \\
0 &
\left[
\begin{matrix}
Y_{dd}^2 & Y_{dq}^2 \\
Y_{qd}^2 & Y_{qq}^2
\end{matrix}
\right] & \cdots & 0 \\
\vdots & \vdots & \ddots & \vdots \\
0 & 0 & \cdots &
\left[
\begin{matrix}
Y_{dd}^N & Y_{dq}^N \\
Y_{qd}^N & Y_{qq}^N
\end{matrix}
\right]
\end{bmatrix} 
\end{equation}

The RoG subsystem is represented as a $(2N \times 2N)$ impedance matrix which is a full matrix given by:
\begin{equation}
\label{eq:Zgrid}
\setlength{\arraycolsep}{3pt}
\resizebox{\columnwidth}{!}{$
\mathbf{Z}_G(s)=
\left[
\begin{array}{cccc}
\Zblk{1}{1} & \Zblk{1}{2} & \cdots & \Zblk{1}{N} \\
\Zblk{2}{1} & \Zblk{2}{2} & \cdots & \Zblk{2}{N} \\
\vdots      & \vdots      & \ddots & \vdots      \\
\Zblk{N}{1} & \Zblk{N}{2} & \cdots & \Zblk{N}{N}
\end{array}
\right]
$}
\end{equation}

Note that Laplace notation $(s)$ is dropped for all the elements of $\mathbf{Y_I}(s)$ and $\mathbf{Z_G}(s)$ for brevity. In practice, $\mathbf{Y_I}(s)$ and $\mathbf{Z_G}(s)$ are obtained using dynamic frequency scans \cite{GIST2023Shah} to estimate their frequency response, followed by vector fitting \cite{vectorfitting} to derive the corresponding transfer function.

\begin{figure}[!t]
\centering
\includegraphics[width=1.\linewidth]{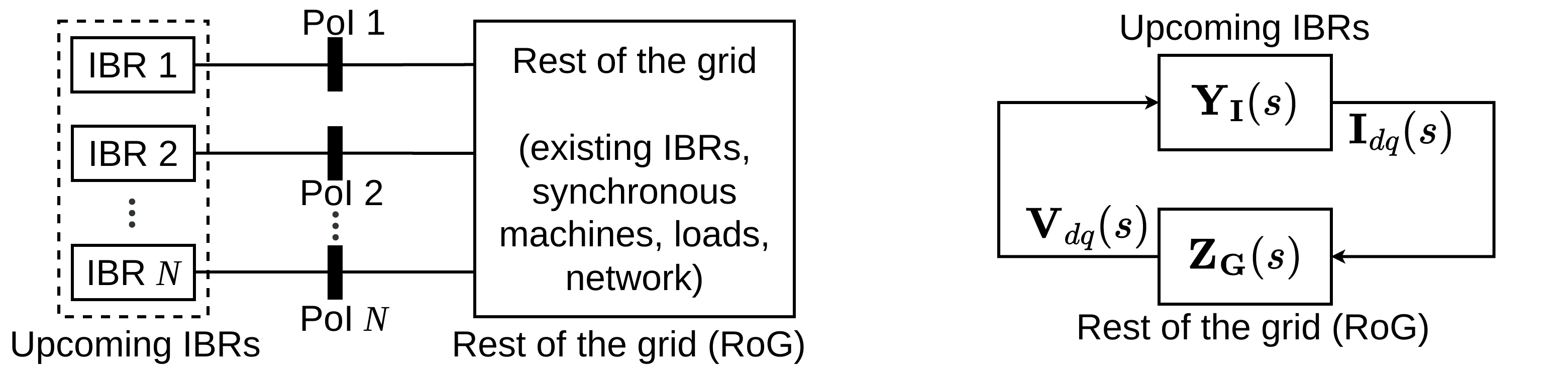}
\caption{Left subfigure: A power system with $N$ upcoming IBR plants to be connected to the rest of the grid (RoG) comprising existing IBR plants, synchronous machines, loads, and the network. The busbars represent the points of interconnection (PoI) of the upcoming IBR plants. Right subfigure: Positive feedback structure with two subsystems - upcoming IBR admittance $\mathbf{Y_I}(s)$ and rest of the grid (RoG) impedance $\mathbf{Z_G}(s)$}
\label{fig:IBRs_RoG}
\end{figure}

As shown in Fig.~\ref{fig:IBRs_RoG}, the IBR-RoG closed-loop setup is generally modeled with positive feedback, while the diagonal dominance theory in Section~\ref{sec:DD_theory} is formulated for the negative feedback structure shown in Fig.~\ref{fig:mimoloop}. Therefore, to cast the IBR-RoG setup into the same negative feedback form, we define
% As shown in Fig.~\ref{fig:IBRs_RoG}, the IBR-RoG interaction is inherently of positive feedback form, while the diagonal dominance theory in Section \ref{sec:DD_theory} is formulated for the negative feedback structure shown in Fig.~\ref{fig:mimoloop}. Therefore, to cast the IBR-RoG interaction into the same negative feedback form, we define

% The feedback structure with the two subsystems $\mathbf{Y_I}(s)$ and $\mathbf{Z_G}(s)$ are shown in Fig.~\ref{fig:IBRs_RoG}. Referring to Fig.~\ref{fig:mimoloop}, $\mathbf{Y_I}(s)$ and $-\mathbf{Z_G}(s)$ correspond to $H_1$ and $H_2$, respectively. The negative sign reflects the fact that IBR-RoG interaction is inherently positive feedback, while diagonal dominance theory in Section \ref{sec:DD_theory} is formulated for negative feedback. Accordingly, define
\begin{equation}
\mathbf{L}(s) = -\mathbf{Y_I}(s) \mathbf{Z_G}(s), 
\label{eq:loop_blocks}
\end{equation}
which is used hereafter for brevity. Partition $\mathbf{L}(s)$ into $2\times 2$ $dq$-domain blocks as
\begin{equation*}
\mathbf{L}^{(i,j)}(s) = -\mathbf{Y_I}^{i}(s)\mathbf{Z_G}^{(i,j)}(s)
=
\begin{bmatrix}
L_{dd}^{ij}(s) & L_{dq}^{ij}(s) \\
L_{qd}^{ij}(s) & L_{qq}^{ij}(s)
\end{bmatrix}.
\label{eq:loop_block}
\end{equation*}

Based on Theorem~\ref{DDsuf} and similar feedback setup in Fig. \ref{fig:mimoloop}, the closed-loop stability (and a minimum decay rate) of the overall interconnected power system in Fig.~\ref{fig:IBRs_RoG}(b) can be certified if 
\begin{equation}
 \mathbf{M}(s) = \mathbf{\mu}\mathbf{I}+\mathbf{L}(s) \textrm{ is SDD }  \;    \forall \, s
\in \partial \mathbb{C}_{<\alpha}, \; \forall \, \mu \in [1,+\infty].
\label{eq:strict_diagonal_dominance_condition}
\end{equation}

Due to the block diagonal structure of $\mathbf{Y_I}(s)$ in \eqref{eq:Y_block}, the diagonal dominance condition admits a fully decentralized interpretation.
To make this explicit, $\mathbf{M}(s) \in \mathbb{C}^{2N\times2N}$ in \eqref{eq:strict_diagonal_dominance_condition} is partitioned into $2\times 2$ blocks, given by:
\begin{subequations}\label{eq:H_block_dq_expansion_decentralized}
\begingroup
\setlength{\arraycolsep}{2pt}
\renewcommand{\arraystretch}{1.0}
\begin{align}
\mathbf{M}^{(i,i)}(s)
&=
\left[
\begin{array}{@{}cc@{}}
\mu+L_{dd}^{ii}(s) & L_{dq}^{ii}(s)
\\[0.4ex]
L_{qd}^{ii}(s) & \mu+L_{qq}^{ii}(s)
\end{array}
\right],
\label{eq:H_block_dq_expansion_decentralized_1}
\\
\mathbf{M}^{(i,j)}(s)
&=
\left[
\begin{array}{@{}cc@{}}
L_{dd}^{ij}(s) & L_{dq}^{ij}(s)
\\[0.4ex]
L_{qd}^{ij}(s) & L_{qq}^{ij}(s)
\end{array}
\right].
\label{eq:H_block_dq_expansion_decentralized_2}
\end{align}
\endgroup
\end{subequations}
Then,
\begin{equation}
\mathbf{M}^{(i,-i)}
=
\big[
\mathbf{M}^{(i,1)},\,
\ldots,\,
\mathbf{M}^{(i,i-1)},\,
\mathbf{M}^{(i,i+1)},\,
\ldots,\,
\mathbf{M}^{(i,N)}
\big].
\label{eq:M_i_minus_i}
\end{equation}
Note that Laplace notation $(s)$ in \eqref{eq:M_i_minus_i} is dropped for brevity. Notably, $\mathbf{M}^{(i,i)}$ in \eqref{eq:H_block_dq_expansion_decentralized_1} and $\mathbf{M}^{(i,-i)}$ in \eqref{eq:M_i_minus_i} correspond to the blocks $M_{i,i}$ and $M_{i,-i}$ in \eqref{relaxedDD}.
% Accordingly, the block row $i$ corresponding to the upcoming IBR at bus $i$ is composed solely of its own admittance $\mathbf{Y_I}^i(s)$ and only those elements $\mathbf{Z_G}^{(i,j)}(s)$ of the RoG impedance that are connected to bus $i$.

Based on Theorem~\ref{blockDD}, we define the block diagonal dominance (BDD) margin of the $i$-th IBR:
\begin{equation}
\delta^i(s)
=
\inf_{\mu \in [1,+\infty)}
\left[
1-
\left\|
\left(\mathbf{M}^{(i,i)}\right)^{-1}\mathbf{M}^{(i,-i)}
\right\|_{\infty}
\right].
\label{eq:Delta_i}
\end{equation}
% \begin{equation}
% \delta^i(s)
% =
% 1-\left\|\left(\mathbf{M}^{(i,i)}(s)\right)^{-1}\mathbf{M}^{(i,-i)}(s)\right\|_{\infty}.
% \label{eq:Delta_i}
% \end{equation}
The closed-loop stability and specified decay rate of the overall system in Fig.~\ref{fig:IBRs_RoG}(b) can be certified if for all upcoming IBRs, $i \in \{1,\ldots,N \}$, the following BDD-based condition holds:
\begin{equation}
\delta^i(s) > 0,   \;    \forall \, s
\in \partial \mathbb{C}_{<\alpha}.
\label{eq:bdd_condition_power}
\end{equation}
% blocks on the $i$th row of $\mathbf{L}(s)$ are constructed solely from the local admittance dynamics $\mathbf{Y}_c^i(s)$ of the $i$th IBR and the coupling grid impedances between its PoI and the PoIs of other upcoming IBRs. 

% Under this formulation, the $M$ in \eqref{bmat} is $\mu \mathbf{I}+\mathbf{H}$
% Under this formulation, the diagonal dominance condition reduces to verifying that, for each IBR
% $i \in \{1,\ldots,N\}$ and for each axis $k \in \{D,Q\}$,
% \begin{equation}
% \left|
% 1 + H_{kk}^{(i,i)}(s)
% \right|
% >
% \sum_{\substack{j=1 \\ j \neq i}}^{N}
% \;
% \sum_{m \in \{D,Q\}}
% \left|
% H_{km}^{(i,j)}(s)
% \right|,
% \quad \forall s
% \in \partial \mathbb{C}_{<\alpha}. 
% \label{eq:fully_decentralized_dd_condition}
% \end{equation}
According to \eqref{eq:loop_blocks}, \eqref{eq:H_block_dq_expansion_decentralized}--\eqref{eq:Delta_i}, the condition \eqref{eq:bdd_condition_power} can be assessed in a decentralized manner, requiring only the local admittance $\mathbf{Y_I}^i(s)$ from each upcoming IBR, and the grid-impedance dynamics at that PoI, corresponding to each row in \eqref{eq:Zgrid}.
No information about other IBRs needs to be disclosed. Consequently, the vendor or developer of each IBR can independently verify the condition without knowledge of the other IBRs in the system. If each upcoming IBR satisfies the condition at its respective PoI, the overall system is guaranteed to be stable with the specified minimum decay rate.

In the following Section \ref{sec:single-ibr case study}, the diagonal dominance criterion in Section \ref{sec:DD_theory} is applied to a single IBR connected to a Thevenin-equivalent grid, with the emphasis on its ability to certify closed-loop stability and a prescribed minimum decay rate. It should be noted that with only one IBR, this case serves mainly to illustrate the certification procedure itself. The full block coupling mechanism underlying Theorem~\ref{blockDD} and the derived condition~\eqref{eq:bdd_condition_power}, arises only in the multi-IBR setting presented later in Section~\ref{multi-ibr case study}. 
% \gl{In the following Section} \ref{sec:single-ibr case study}, \gl{a preparatory example of a single IBR connected to a Thevenin-equivalent grid is presented to illustrate, how the diagonal dominance criterion can be used to certify closed-loop stability and prescribed minimum decay rate. It should be noted, however, that with only one IBR, this simple case serves mainly to illustrate the certification procedure itself. The full block coupling mechanism underlying Theorem}~\ref{blockDD} \gl{and the derived condition}~\eqref{eq:bdd_condition_power}, \gl{arises only in the multi-IBR setting presented later in Section}~\ref{multi-ibr case study}. 
% Next, the proposed decentralized condition is evaluated on a single-IBR system connected to a Thevenin-equivalent grid, with emphasis on its ability to certify closed-loop stability and a prescribed minimum decay rate, as well as on its degree of conservativeness. Due to simplistic single IBR setup, the BDD-based condition becomes equivalent to the strict diagonal dominance (SDD)-based condition, as discussed next.

\section{Stability and Decay Rate Certification in Single-IBR Power System}\label{sec:single-ibr case study}
% In this section, the proposed decentralized condition is first illustrated on a simple single-IBR benchmark connected to a Thevenin-equivalent grid, as shown in Fig.~\ref{fig:s_ibr_thev}. 
We start with a single-IBR connected to a Thevenin-equivalent grid as shown in Fig.~\ref{fig:s_ibr_thev}. 
% \gl{to demonstrate how the frequency-domain certification, based on diagonal dominance criterion, of closed-loop stability and the prescribed minimum decay rate works in the simplest setting.} 
The IBR operates in grid-following (GFL) mode regulating both active power and voltage magnitude at the point of interconnection (PoI) to 1 pu. The Thevenin equivalent impedance is characterized by $\textrm{SCR}=3$ and $X/R=10$. For the given operating operating point, $2\times2~dq$ MIMO grid impedance $Z_G(s)$ and IBR admittance $Y_I(s)$ transfer functions are obtained using dynamic frequency scans followed by vector fitting \cite{aemo2025frequency, GIST2023Shah, vectorfitting}. In this setting, since there is only one upcoming IBR, the multi-block structure underlying \eqref{eq:Y_block}--\eqref{eq:loop_blocks} collapses and the block matrix $\mathbf{M}^{(i,-i)}$ of \eqref{eq:M_i_minus_i} is absent. Accordingly, BDD-based condition \eqref{eq:bdd_condition_power} is not a meaningful nontrivial condition in this setting. Instead, the certification of closed-loop stability and decay rate is carried out through the corresponding scalar variant, namely the strict diagonal dominance (SDD)-based condition, as derived in \eqref{eq:local_sdd_stab_condition}. 
For this single-IBR setting, \eqref{eq:loop_blocks} can be written explicitly as follows:
% In this section, the proposed decentralized condition is first illustrated on a simple single-IBR benchmark connected to a Thevenin-equivalent grid, as shown in Fig.~\ref{fig:s_ibr_thev}. The IBR operates in grid-following (GFL) mode regulating both active power and voltage magnitude at the point of interconnection (PoI) to 1 pu. The Thevenin equivalent impedance is characterized by a $\textrm{SCR}=3$ and $X/R=10$. For the given operating operating point, $2\times2~dq$ MIMO grid impedance $Z_G(s)$ and IBR admittance $Y_I(s)$ transfer functions are obtained using dynamic frequency scans followed by vector fitting. 
% In this case, since there is only one upcoming IBR, the multi-block structure in \eqref{eq:Y_block}--\eqref{eq:loop_blocks}, and \eqref{eq:M_i_minus_i} collapses. As a result, BDD boils down to strict diagonal dominance (SDD). For this system, \eqref{eq:loop_blocks} can be written explicitly as follows:
\begin{equation*}
\mathbf{L}(s) = -\mathbf{Y_I}(s)\mathbf{Z_G}(s)
=
\begin{bmatrix}
L_{dd}(s) & L_{dq}(s) \\
L_{qd}(s) & L_{qq}(s)
\end{bmatrix}.
\label{eq:loop_gain}
\end{equation*}

Following Theorem~\ref{DDsuf}, we define
% \begin{subequations*}
% \label{eq:loop_gain_dq}
\begin{align*}
m_{dd}(s)
&=
\begin{cases}
\operatorname{Im}\!\left(L_{dd}(s)\right), & \textrm{if } \operatorname{Re}\!\left(L_{dd}(s)\right)\leq -1, \\[0.6ex]
\left|1+L_{dd}(s)\right|, & \textrm{if } \operatorname{Re}\!\left(L_{dd}(s)\right)> -1 ,
\end{cases}
% \label{eq:loop_gain_dd}
\\[1ex]
m_{qq}(s)
&=
\begin{cases}
\operatorname{Im}\!\left(L_{qq}(s)\right), & \textrm{if } \operatorname{Re}\!\left(L_{qq}(s)\right)\leq -1, \\[0.6ex]
\left|1+L_{qq}(s)\right|, & \textrm{if } \operatorname{Re}\!\left(L_{qq}(s)\right)> -1 ,
\end{cases}
% \label{eq:loop_gain_qq}
\end{align*}
% \end{subequations*}
which represent the effective diagonal contributions entering the SDD assessment. Based on these quantities, the corresponding local margins can be defined by comparing the effective diagonal contribution with the magnitude of the associated off-diagonal term on each row, namely
\begin{subequations}
\begin{align}
\delta_{d}(s) 
&= m_{dd}(s) - \left|L_{dq}(s)\right|,
\label{eq:delta_id} \\
\delta_{q}(s) 
&= m_{qq}(s) - \left|L_{qd}(s)\right|.
\label{eq:delta_iq}
\end{align}
\end{subequations}

According to Theorem~\ref{DDsuf}, the condition in \eqref{eq:strict_diagonal_dominance_condition} holds only when both \eqref{eq:delta_id} and \eqref{eq:delta_iq} are strictly positive simultaneously. We define the following strict diagonal dominance (SDD) margin of the upcoming IBR:
\begin{equation}
\delta(s)
= \min\!\big(\delta_{d}(s),\,\delta_{q}(s)\big),
\label{eq:delta_i}
\end{equation}
% as a complement to \eqref{eq:Delta_i}.
which combines the two inequalities in a single quantity. The closed-loop stability and the specified decay rate of the system in Fig.~\ref{fig:s_ibr_thev} can be certified if the following SDD-based condition holds:
\begin{equation}
\delta(s) > 0,   \;    \forall \, s
\in \partial \mathbb{C}_{<\alpha}.
\label{eq:local_sdd_stab_condition}
\end{equation}

% Thus, the local minimum decay-rate certificate condition (equivalent to BDD for this system) is 
% \begin{equation}
% \delta_{dq}(s_{\alpha}) > 0, \qquad \forall \omega \in \Omega.
% \label{eq:local_sdd_condition}
% \end{equation}
% Similarly, the condition to certify stability becomes 
% \begin{equation}
% \delta_{dq}(j\omega) > 0, \qquad \forall \omega \in \Omega.
% \label{eq:local_sdd_stab_condition}
% \end{equation}
For convenience, the complex variable on the boundary $\partial \mathbb{C}_{< \alpha}$ of the certification region is from now on denoted by
\begin{equation}
s_{\alpha} = \alpha + j\omega, \qquad \omega \in \Omega.
\label{eq:s_alpha}
\end{equation}
where $\alpha \le 0$ is the certification threshold, and $\Omega$ denotes the frequency set of interest. The variable $s$ in \eqref{eq:delta_id}, \eqref{eq:delta_iq} and \eqref{eq:delta_i} can be replaced by $s_{\alpha}$.
% Equivalently, $s_{\alpha}\in \partial \mathbb{C}_{<\alpha}$, where
% \[
% \partial \mathbb{C}_{<\alpha}
% =
% \{\, s:\exists\,\omega\in\mathbb{R},\ s=\alpha+j\omega \,\}.
% \]
% With this notation, the conditions in the corresponding theorem can be evaluated directly on $s_{\alpha}$. 

In particular, when $\alpha=0$, \eqref{eq:s_alpha} becomes $s_{\alpha}=j\omega$,
which corresponds to the closed-loop stability certification case. 
% When $\alpha<0$, $s_{\alpha}=\alpha+j\omega$ is used to certify the decay rate requirement $\Re(\lambda)<\alpha$.

% Replacing $s$ in \eqref{eq:local_sdd_stab_condition} with $s_{\alpha}$ yields the following equivalent condition:
% \begin{equation}
% \delta(s_{\alpha}) > 0, \qquad \omega \in \Omega.
% \label{eq:local_sdd_stab_condition_s_alpha}
% \end{equation}

% This setup is used to demonstrate how the diagonal dominance condition can be employed both to certify closed-loop stability and to enforce a specific performance prescribed in terms of minimum decay rate. It also highlights the conservative nature of the compliance test relative to the actual onset of instability or inadequate decay rate.

% \subsubsection{Certificate of stability}
\subsection{Certificate of stability}
% For stability certificate, the frequency response of $Z_G(s)$ and $Y_I(s)$ is evaluated at $s=j\omega$. 
% and is shown in Fig. \ref{fig:bode_s_ibr_thev}.
%If $L(j\omega)$ denotes the loop gain, i.e., $L(j\omega) = Y_I(j\omega)Z_G(j\omega)$, then the BDD condition in its reduced form (equivalent to corresponding SDD) is given below:
% \begin{equation}
% \begin{aligned}
% \left| 1+L_{dd}(j\omega) \right| > \left| L_{dq}(j\omega) \right|,\\ 
% \quad
% \left| 1+L_{qq}(j\omega) \right| > \left| L_{qd}(j\omega) \right|,
% \label{eq:dd_conditions}
% \end{aligned}
% \end{equation}
% for $\omega$ of interest.

\begin{figure}[t!]
    \centering
    \includegraphics[width=1\linewidth]{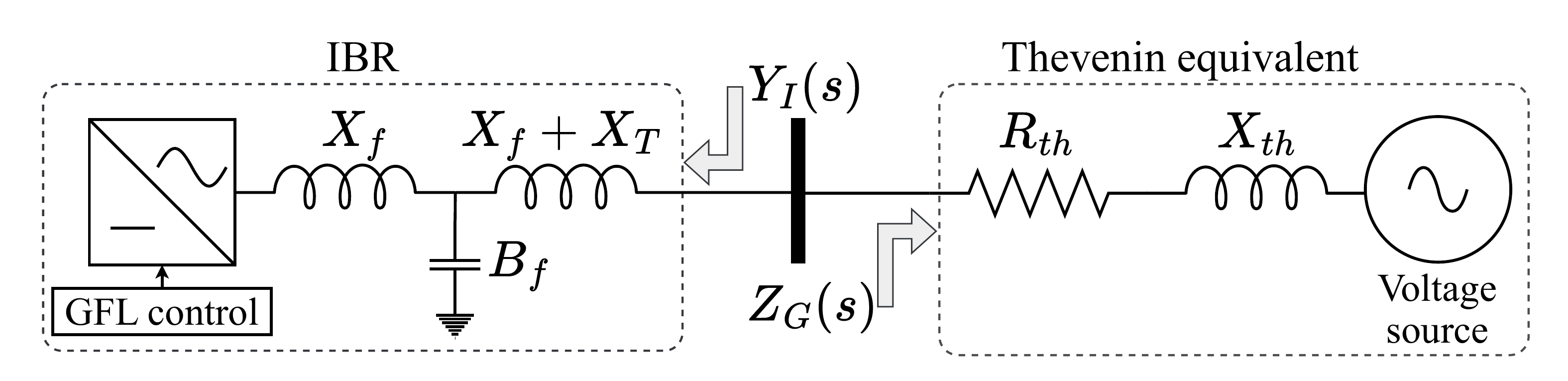}
    \caption{Circuit diagram of single IBR connected to a Thevenin equivalent of a power system. $Z_G(s)$ and $Y_I(s)$ represent impedance and admittance models of the grid and IBR, respectively.}
    \label{fig:s_ibr_thev}
\end{figure}

% \begin{figure}[t!]
%     \centering
%     \includegraphics[width=1\linewidth]{New_Figures/fig_bode_YZ.pdf}
%     \caption{Frequency response of IBR admittance $Y_I(s)$, and Thevenin equivalent impedance, $Z_G(s)$, evaluated at $s=j\omega$, for an IBR in GFL mode at $\textrm{SCR}=3$ and $X/R=10$.}
%     \label{fig:bode_s_ibr_thev}
% \end{figure}

 The stability condition of \eqref{eq:local_sdd_stab_condition} is applicable when both subsystems are standalone stable. This is consistent with practical grid-connection compliance requirements, under which a standalone IBR must be stable. Also, the Thevenin equivalent is passive and therefore intrinsically stable.

Fig.~\ref{fig:DD_thev_alpha0} illustrates the assessment of condition \eqref{eq:local_sdd_stab_condition} by plotting  $\delta_d(j\omega)$ and $\delta_q(j\omega)$ across the frequency of intersts (1 Hz to 500 Hz required by NESO \cite{NESO2025Guidance}) for different SCRs values ranging from 3 to 1. From SCR $=3$ to $1.7$, both quantities remain positive over the entire frequency range. 
% This indicates that the diagonal dominance condition is satisfied on both the $d$- and $q$-axes.Consequently, the proposed DD-based criterion certifies the stability of the closed-loop system for SCR $>1.7$. 
Consequently, the SDD-based condition of \eqref{eq:local_sdd_stab_condition} is satisfied, certifying the stability of the closed-loop system for SCR $>1.7$. 
However, with a decrease in SCR, a dip gradually appears in the $\delta_q(j\omega)$ curve at low frequencies. The results show that the local margin \eqref{eq:delta_iq} is no longer strictly positive throughout the frequency range when the SCR falls below $1.7$, indicating that the SDD-based condition is not satisfied anymore and therefore closed-loop stability cannot be guaranteed.
\begin{figure}[t!]
    \centering
    \includegraphics[width=1\linewidth]{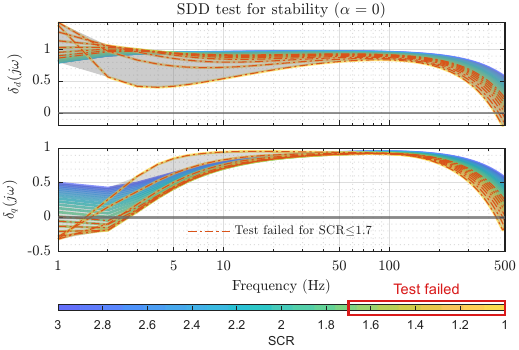}
    % \caption{Diagonal-dominance based stability margins, $\delta_d(j\omega)$ (top) and $\delta_q(j\omega)$ (bottom), for a single IBR connected to a Thevenin grid equivalent with $\alpha=0$. Solid curves are color-coded by SCR (grid strength falls from blue to yellow). Dashed red traces highlight operating points that fail the diagonal-dominance test (here, for $\mathrm{SCR}\leq 1.7$).}
    \caption{Local SDD margins $\delta_d(j\omega)$ (top) and $\delta_q(j\omega)$ (bottom), for a single IBR connected to a Thevenin grid equivalent. All traces are color-coded by SCR (grid strength falls from blue to yellow). Dashed red traces highlight operating points that fail the SDD test (here, for $\mathrm{SCR}\leq 1.7$).}
    \label{fig:DD_thev_alpha0}
\end{figure}

% the following figure is in png because when exported as eps the shading has problem with transparency/opaqueness and when exported with pdf, it crops from the right. so this is a png with 900 dpi
\begin{figure}[t!]
    \centering
    \includegraphics[width=1\linewidth]{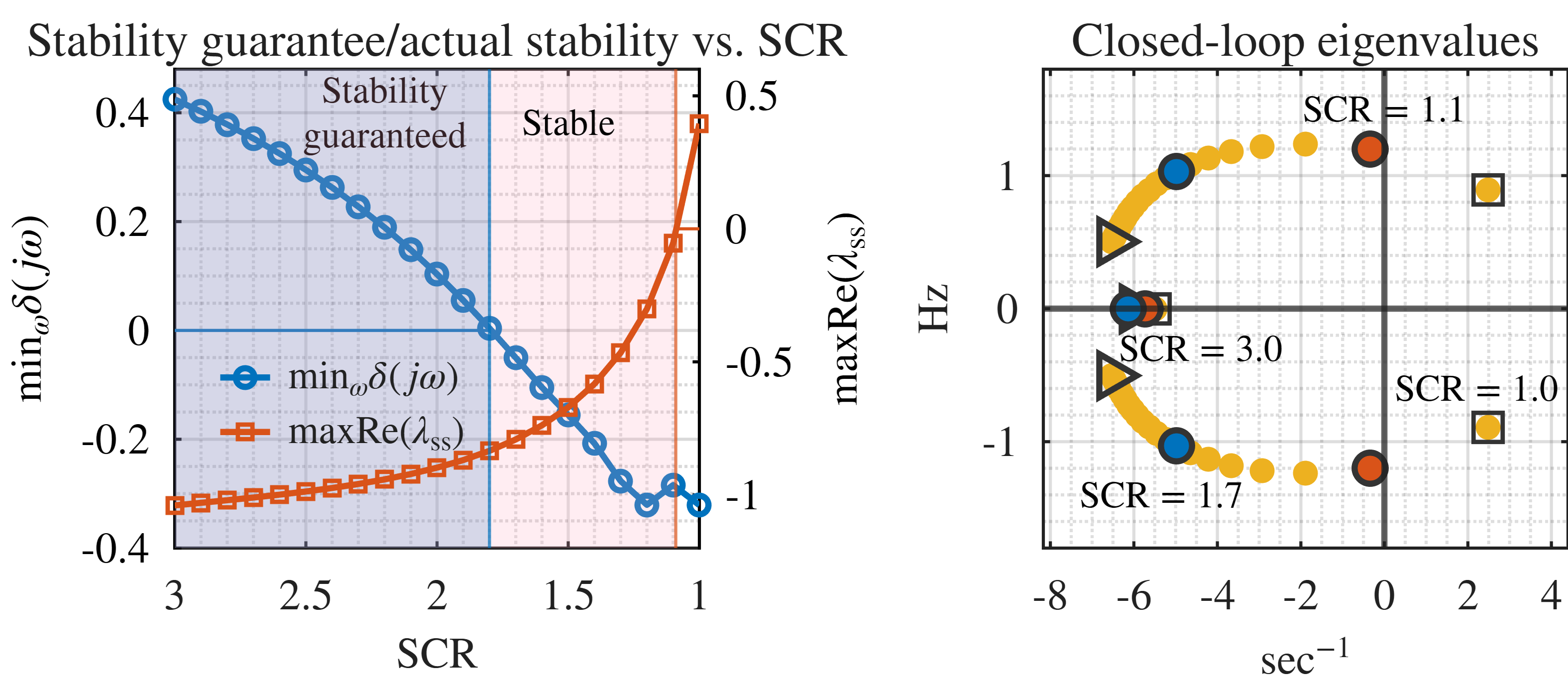}
    \caption{Left: minimum SDD margin $\min_{\omega}\delta(j\omega)$ over the frequency range of interest (blue, left axis) and the real part of least-damped sub-synchronous (SS) mode $\max\,\mathrm{Re}(\lambda_{\mathrm{ss}})$ (red, right axis). Blue shaded region denotes guaranteed stable SCR range and red shaded region shows stability range from eigenvalue analysis. Right: closed-loop SS modes for varying SCR from 3 to 1. Blue circles mark eigenvalues at stability guarantee threshold, red circles mark eigenvalues at actual stability threshold. SCR sweep start has triangle marker and the end has square marker.}
    \label{fig:margin_ev_thev_alpha0}
\end{figure}

To validate the stability certificate, Fig.~\ref{fig:margin_ev_thev_alpha0} compares the SDD margin defined in \eqref{eq:delta_i} with the actual small-signal stability boundary obtained from eigenvalue analysis as SCR decreases. The left subfigure plots two indicators: the minimum SDD margin over the frequency range of interest, $\min_{\omega}\delta(j\omega)$, and the real part of the least-damped sub-synchronous (SS) mode, $\max\,\mathrm{Re}(\lambda_{\mathrm{ss}})$. The SDD-based certificate holds when $\min_{\omega}\delta(j\omega)>0$ (for SCR $>1.7$). This is confirmed by the closed-loop eigenvalues shown in the right subfigure of Fig.~\ref{fig:margin_ev_thev_alpha0}, where all eigenvalues lie strictly in the left-half complex plane. As the SCR decreases, the SDD margin reduces monotonically and reaches zero at SCR = 1.7, prior to any eigenvalue crossing into the right-half plane. The system becomes oscillatory when SCR falls below approximately 1.1, where a poorly damped mode reaches the imaginary axis. As illustrated in Fig.~\ref{fig:margin_ev_thev_alpha0}, the stability region certified by the SDD-based condition is a subset of the true stability region identified by eigenvalue analysis, thereby demonstrating its conservative nature.

% The SDD-based certificate is valid as long as $\min_{\omega}\delta(j\omega)>0$, which occurs for relatively larger SCR. As the grid weakens, this minimum decreases monotonically and reaches zero before any SS eigenvalue crosses into the right-half plane. This shows that the proposed condition loses its guarantee earlier than the actual onset of instability, thereby demonstrating its conservative nature. Nevertheless, throughout the region where the SDD condition certifies stability, the actual closed-loop eigenvalues (shown on the right in Fig.~\ref{fig:margin_ev_thev_alpha0}) confirm that the system is indeed stable. When SCR $>1.7$, all eigenvalues lie strictly in the left-half complex plane. This observation is consistent with the stability estimation obtained from the SDD-based analysis. The system becomes oscillatory when SCR falls below $1.1$, where a poorly damped mode reaches the imaginary axis. Therefore, compared with the true stability margin obtained from eigenvalue analysis, the SDD condition yields a stability threshold of SCR $\approx 1.7$, whereas the actual stability limit occurs at SCR $\approx 1.1$.

% \subsubsection{Certificate of minimum decay-rate}
\subsection{Certificate of minimum decay rate}
% In practical power systems, grid-connection requirements typically extend beyond closed-loop stability and impose minimum decay rate constraints to avoid poorly damped oscillatory behavior. In particular, an IBR to be connected should not introduce poorly damped oscillations once integrated into the system. Therefore, beyond certifying closed-loop stability, the SDD condition can be extended to certify a prescribed minimum decay rate, corresponding to $\alpha\leq0$ in (\ref{eq:s_alpha}).

% Specifically, a desired decay rate margin can be enforced by applying the SDD-based analysis to a shifted open-loop transfer function. Based on Theorem \ref{DDsuf}, if the closed-loop system is required to satisfy a minimum decay rate requirement such that the real parts of all eigenvalues remain smaller than a prescribed threshold $\alpha$, then the dynamics of both the upcoming IBR and the existing grid are evaluated after a rightward shift of $-\alpha$ along the real axis in the complex plane. To certify a prescribed minimum decay rate, the SDD condition is assessed on the shifted vertical contour in (\ref{eq:s_alpha}), rather than on the imaginary axis. 

In practical power systems, minimum decay-rate requirements are often imposed in addition to mere stability. According to Theorem~\ref{DDsuf}, the same SDD test can certify a prescribed decay rate by evaluating the converter and grid dynamics along the shifted contour \(s_\alpha=\alpha+j\omega\), with \(\alpha\leq 0\).

%In practical power systems, grid-connection requirements extend beyond stability and impose minimum decay rate constraints to avoid introducing poorly damped oscillatory behavior. Accordingly, the SDD condition can be extended to certify a prescribed decay rate, corresponding to $\alpha\leq0$ in (\ref{eq:s_alpha}). Based on Theorem \ref{DDsuf}, this is achieved by applying the SDD-based analysis to a shifted open-loop transfer function, where both the IBR and network dynamics are evaluated along a vertical contour shifted by $-\alpha$ in the complex plane. [discarded as it was more detailed]

% Here, $\alpha=-0.8 ~\textrm{sec}^{-1}$, which corresponds to all closed-loop poles to lie to the left of \(\mathrm{Re}(s)=-0.8\).
%Using the standard \(2\%\) settling-time approximation, this is equivalent to a settling-time requirement of approximately \approx 3.18~\text{s}.
% 
% Consequently, the loop gain becomes \(L(s_\alpha)=Y_I(s_\alpha)Z_G(s_\alpha)\).
% 
% then the reduced DD condition becomes
% \begin{equation}
% \begin{aligned}
% \left| 1+L_{dd}(s_\alpha) \right| > \left| L_{dq}(s_\alpha) \right|,\\
% \quad
% \left| 1+L_{qq}(s_\alpha) \right| > \left| L_{qd}(s_\alpha) \right|,
% \end{aligned}
% \label{eq:dd_conditions_alpha}
% \end{equation}
% for frequencies of interest.

Here, \(\alpha=-0.8~\textrm{sec}^{-1}\), corresponding to a 5~sec settling-time limit under the \(2\%\) settling-time criterion based on the dominant pole. Fig.~\ref{fig:DD_thev_alpha_mp2} shows \(\delta_d(s_\alpha)\) and \(\delta_q(s_\alpha)\) over the same frequency and SCR ranges. The margins deteriorate at low frequencies and lose positivity when SCR falls below \(1.9\). Consequently, condition \eqref{eq:local_sdd_stab_condition} is violated, and the prescribed decay rate can no longer be certified.

%Here, the case of $\alpha=-0.8~\textrm{$sec^{-1}$}$, corresponding to a maximum settling time of 5$s$ (2$\%$ settling band), is discussed. Fig.~\ref{fig:DD_thev_alpha_mp2} illustrates this compliance test by plotting
%\(\delta_d(s_\alpha)\) %and %\(\delta_q(s_\alpha)\) %over the same frequency range of interest and SCR range.%The minimum decay-rate certificate is valid only if both margins remain positive over the entire frequency range. The results show that this condition is satisfied for sufficiently strong grids, but progressively weakens as the SCR decreases.
% With the goal to meet the decay-rate of $\alpha=-0.8~s^{-1}$, the minimum margin becomes negative when the SCR falls below 1.9, 
%It shows that the local margin \eqref{eq:delta_iq} deteriorates at low frequencies and is no longer strictly positive throughout the frequency range when the SCR falls below $1.9$, indicating that the prescribed decay rate requirement can no longer be certified from the SDD test results. 
% As in the stability case, the loss of certificate is driven primarily by the deterioration of the \(q\)-axis margin at low frequencies. [discarded and replaced by smaller version above]

Fig.~\ref{fig:margin_ev_thev_alpha_mp2} (left subfigure) compares the SDD-based certificate with the actual decay-rate boundary from eigenvalue analysis. The certificate remains valid for \(\mathrm{SCR}>1.9\), where \(\min_{\omega}\delta(s_\alpha)>0\). As SCR decreases, the margin reaches zero before the critical SSO mode crosses \(\mathrm{Re}(s)=\alpha\), confirming conservativeness. Thus, though the SDD-based condition is violated before the actual minimum decay-rate boundary, it still guarantees the prescribed minimum decay rate within the region, as verified by the closed-loop eigenvalue analysis shown on the right subfigure of Fig.~\ref{fig:margin_ev_thev_alpha_mp2}.

\begin{figure}[t!]
    \centering
    \includegraphics[width=1\linewidth]{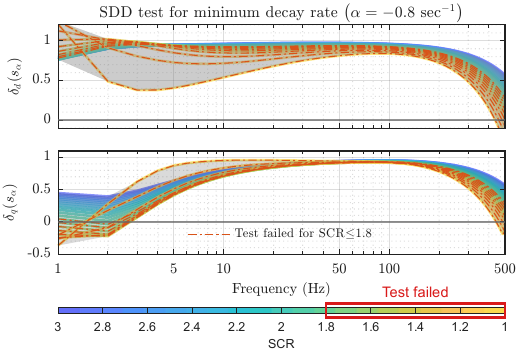}
    \caption{Local SDD margins, $\delta_d(j\omega)$ (top) and $\delta_q(j\omega)$ (bottom), for a single IBR connected to a Thevenin grid equivalent with $\alpha=-0.8~\textrm{sec}^{-1}$. All traces are color-coded by SCR (grid strength falls from blue to yellow). Dashed red traces highlight operating points that fail the SDD test (here, for $\mathrm{SCR}\leq 1.8$).}
    \label{fig:DD_thev_alpha_mp2}
\end{figure}

% the following figure is in png because when exported as eps the shading has problem with transparency/opaqueness and when exported with pdf, it crops from the right. so this is a png with 900 dpi
\begin{figure}[t!]
    \centering
    \includegraphics[width=1\linewidth]{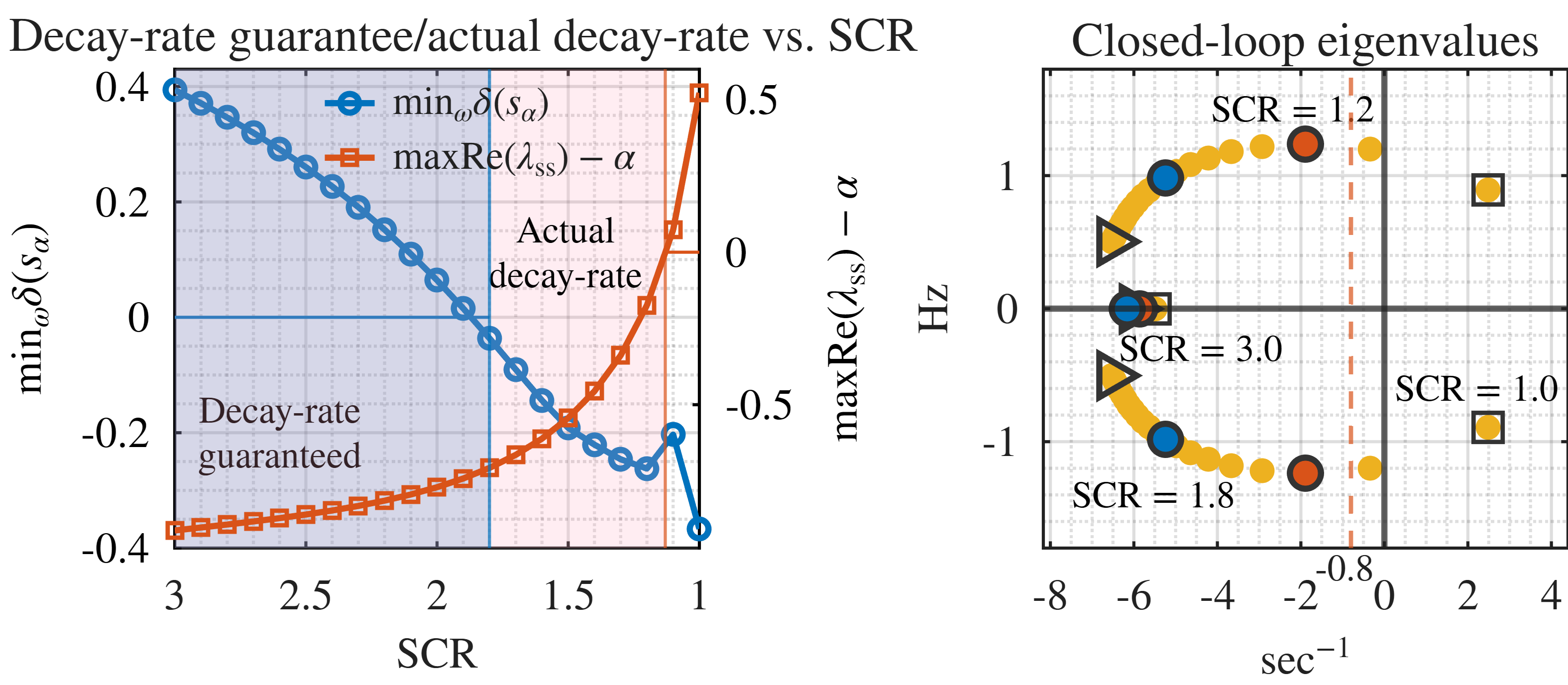}
    \caption{Left: minimum SDD margin $\min_{\omega}\delta(s_\alpha)$ (blue, left axis) and the distance of real part of least-damped sub-synchronous (SS) mode $\max\,\mathrm{Re}(\lambda_{\mathrm{ss}})$ from the prescribed decay rate boundary $\alpha=-0.8 ~\textrm{sec}^{-1}$ (red, right axis). Blue shade denotes region with guaranteed decay rate and red shaded region shows region with actual decay rate. Right: closed-loop SS modes for varying SCR from 3 to 1. Blue circles denote eigenvalues at decay rate guarantee threshold, red circles denote eigenvalues at actual decay rate threshold. SCR sweep start has triangle marker and the end has square marker.}
    \label{fig:margin_ev_thev_alpha_mp2}
\end{figure}

%Overall, the single-IBR Thevenin-equivalent example confirms that the SDD condition provides valid, albeit conservative, certificates of both closed-loop stability and minimum decay rate.

% \section{Decentralized Stability Guarantee in 39-Bus Power System}
\section{Results in Multi-IBR Power System}\label{multi-ibr case study}
\vspace{-1.5em}
While the single-IBR example in Section~\ref{sec:single-ibr case study} demonstrates how the diagonal dominance criterion can be used to certify stability and minimum decay rate, it does not fully reveal the main practical advantage of the proposed framework, namely its decentralized nature. This section therefore extends the analysis to a multi-IBR power system, where the decentralized structure of the proposed BDD-based assessment framework offers guarantees with less conservativeness during the connection compliance process of multiple IBRs.

% While the single-IBR example in Section~\ref{sec:single-ibr case study} demonstrates the use of diagonal dominance criterion, as well as its ability to provide stability and minimum decay rate guarantees, it does not fully reveal the main practical advantage of the method, namely its decentralized nature. In the single-IBR case, the BDD test reduces to its corresponding SDD form, and no inter-IBR coupling and mutual impedance need to be considered. This section therefore extends the analysis to a multi-IBR power system, where the decentralized structure of the proposed BDD-based condition offers guarantees with less conservativeness.

Fig.~\ref{fig:39BusSystem} shows the modified IEEE 39-bus system \cite{ref39bus}, \cite{11082650} considered in this study, where existing IBRs are shown in blue and the upcoming IBRs under compliance assessment are highlighted in orange. The network condition is kept same as in \cite{ref39bus}. The modified active power of each IBR is shown in Table \ref{tab:39Bus_OP}.

Compared with a GFL, a grid-forming inverter (GFM) exhibits dynamic behavior closer to that of a voltage source and generally supports to system stability~\cite{10477540,10508461}. Hence, only those scenarios with upcoming GFLs are shown here. GFMs are included in the RoG.

\begin{figure}[!t]
\centering
\includegraphics[width=0.8\linewidth]{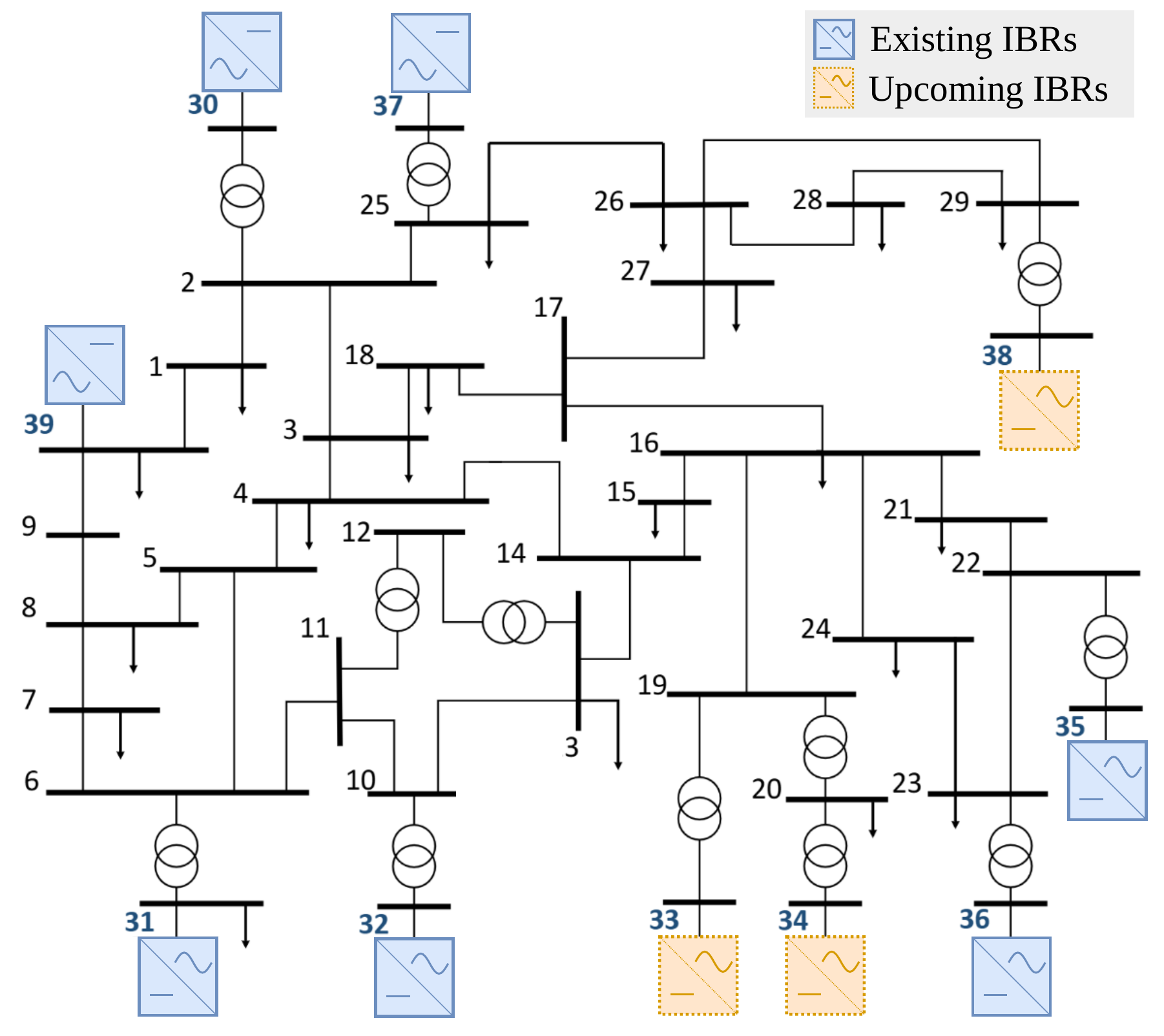}
\caption{A modified IEEE 39-bus system \cite{ref39bus,11082650} with all machines replaced by IBRs. Upcoming IBRs are considered arbitrarily at buses 33, 34 and 38.}
\label{fig:39BusSystem}
\end{figure}

%%%%%%%%%%%%%%%%%%%%%%%%% 39-bus system scenario table %%%%%%%%%%%%
\begin{table}[t!]
    \centering
    \caption{IBRs active power setpoint}
    \vspace{-1em}
    \label{tab:39Bus_OP}
    \renewcommand{\arraystretch}{1.15}
    \setlength{\tabcolsep}{3pt}
    \footnotesize
    \resizebox{\columnwidth}{!}{%
        \begin{tabular}{ccccccccccc}
        \noalign{\hrule height 1.1pt}
        Bus   & 30  & 31  & 32  & 33  & 34  & 35  & 36  & 37  & 38  & 39  \\ \hline
        $P$ (pu) & 4.50 & 12.69 & 6.50 & 7.00 & 5.08 & 5.20 & 4.48 & 4.32 & 3.64 & 8.00 \\
        \noalign{\hrule height 1.1pt}
        \end{tabular}%
    }
\end{table}

\subsection{A True Decentralized Stability Guarantee}
A natural extension of current grid-connection compliance practice using Thevenin equivalent of the grid is to assess each upcoming IBR independently using the impedance spectrum of the rest of the grid (RoG) seen at its prospective PoI and to apply a local Nyquist-type stability check. In this setting, each vendor requires only its own IBR admittance model together with the RoG impedance at its PoI, and can perform the compliance test independently.

However, in a multi-IBR setting, such local Nyquist compliance at each individual PoI does not in general guarantee that the overall closed-loop system will remain stable once all upcoming IBRs are connected simultaneously. The reason is that the local test neglects the mutual coupling among the upcoming IBRs through the rest of the grid.

In contrast to local Nyquist checks, the BDD test explicitly accounts for the row-wise coupling through the the off-diagonal terms of RoG impedance, while still preserving a decentralized structure. Each vendor requires only its own admittance model together with the corresponding row of RoG impedance information. If every upcoming IBR satisfies the proposed BDD-based condition, then the stability, or prescribed performance, of the overall closed-loop multi-IBR system is guaranteed.

As shown next, a case can arise in which all upcoming IBRs individually pass their local Nyquist tests at their respective PoIs, yet the combined closed-loop system is unstable, which is accurately flagged by the proposed BDD test.

Here, we present a case in which all upcoming IBRs (\mbox{IBR33}, \mbox{IBR34}, and \mbox{IBR38}) individually pass local Nyquist-based stability tests at their respective PoIs (Fig.~\ref{fig:Nyquist_3_IBRs}), since $\Delta(j\omega)=\det\!\big(I+L(j\omega)\big)$ does not encircle the origin (top-left, top-right, and bottom-left) and \(\sigma_{\min}(I+L(j\omega))>0\) (bottom-right) throughout the scanned frequency range (1 Hz to 500 Hz\cite{NESO2025Guidance}). However, for the same operating condition, the proposed BDD-based compliance test fails due to IBR33, as shown in Fig.~\ref{fig:BDD_for_Nyquist_EV} (left). In BDD, for full compliance, the margin obtained from \eqref{eq:Delta_i}, should be greater than 0 for all upcoming IBRs according to condition \eqref{eq:bdd_condition_power}.
% Fig.~\ref{fig:BDD_for_Nyquist_EV}(right) verifies that the overall closed-loop system is unstable (highlighted in the zoomed inset). Although failure of the BDD test does not in general imply instability, in this particular case it correctly flags the loss of a stability guarantee, whereas the local Nyquist-based tests remain non-conservative enough to miss the unstable case. This example demonstrates that relying on local Nyquist-based tests to gauge system stability could be misleading. On the other hand, the proposed BDD condition provides a rigorous compliance criterion, and as claimed in the theory, yields a genuine decentralized stability guarantee.
Fig.~\ref{fig:BDD_for_Nyquist_EV} (right) confirms that the overall closed-loop system is actually unstable as evident from the zoomed inset showing eigenvalues on the right-half plane. Due to the sufficient nature of the BDD criterion, this unstable scenario is correctly rejected by the BDD test. It should be noted that failure of the BDD test does not necessarily imply instability, though it turned out to be the case on this occasion. Notably, local Nyquist-based tests for individual IBRs fail to detect this instability in the presence of all three upcoming IBRs. This example demonstrates that relying on local Nyquist-based test could be misleading for multiple IBR connections. In contrast, the proposed BDD-based condition provides a rigorous locally verifiable compliance criterion that guarantees decentralized stability, consistent with the theoretical claims.

% Fig.~\ref{fig:BDD_for_Nyquist_EV}(right) shows that the actual system is indeed unstable (highlighted in the zoomed inset) which is consistent with the conclusion of BDD test but against the findings of Nyquist-based stability tests. This example demonstrates that relying on local Nyquist-based tests to gauge system stability could be misleading. On the other hand, the proposed BDD condition provides a rigorous compliance criterion, and as claimed in the theory, yields a genuine decentralized stability guarantee.

\begin{figure}[!t]
\centering
\includegraphics[width=1\linewidth]{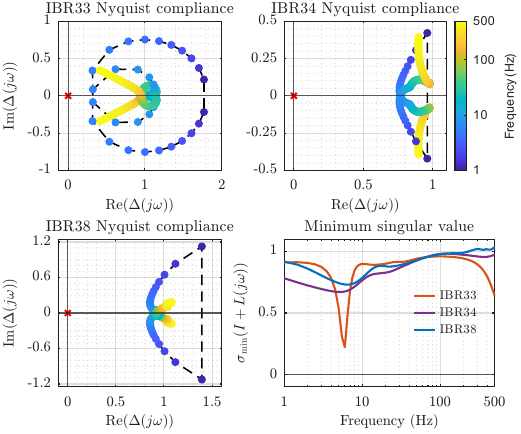}
\caption{Local generalized-Nyquist compliance assessment of the three upcoming IBRs, at their individual prospective PoIs in the 39-bus system. The top-left, top-right, and bottom-left subplots show the generalized Nyquist loci $\Delta(j\omega)$ for IBR33, IBR34, and IBR38, respectively, traced over the Nyquist contour using conjugate symmetry; the points are color-coded by frequency over the scanned range \(1\) Hz to \(500\) Hz, and the red cross marks the critical point at the origin. The bottom-right subplot shows the smallest singular value $\sigma_{\min}\!\big(I+L(j\omega)\big)$ over the scanned frequency range.}
\label{fig:Nyquist_3_IBRs}
\end{figure}

\begin{figure}[!t]
\centering
\includegraphics[width=1\linewidth]{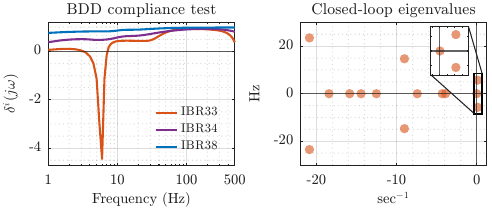}
\caption{BDD-based compliance assessment and actual closed-loop stability. Left: BDD margin \(\delta^{i}(j\omega)\) for each upcoming IBR. Right: overall closed-loop system eigenvalues, showing that a \(5.7\) Hz mode lies in the right-half plane.}
\label{fig:BDD_for_Nyquist_EV}
\end{figure}

It is worth noting that in theory a sequential connection procedure can also be adopted, in which the system operator repeats the network scan considering previous IBR connections and the upcoming IBR reassesses compliance at its PoI. Although feasible in principle, this requires repeated scans (after every connection) and introduces queue dependence, since later IBRs may face progressively stricter compliance requirements as earlier connections consume part of the available stability margin. The other possibility is to conduct a centralized assessment in which the system operator receives the models of all upcoming IBRs and performs a full multi-IBR stability check. This can determine the actual combined-system stability, but it is no longer a decentralized approach, since the compliance assessment is shifted to the system operator and users (vendors/plant developers) cannot test the criteria locally at their end.

\subsection{A Sufficient Condition and Associated Conservativeness}
In this subsection, the focus is on minimum decay rate certification rather than mere stability. Accordingly, the BDD margins are evaluated on the shifted contour \(s_\alpha\) in (\ref{eq:s_alpha}), and the vertical axes in Fig.~\ref{fig:BDD_Sufficiency_EV} are therefore expressed in terms of \(\delta^i(s_\alpha)\), not \(\delta^i(j\omega)\). Here too, $\alpha=-0.8 ~\textrm{sec}^{-1}$.

Fig.~\ref{fig:BDD_Sufficiency_EV} illustrates the sufficient nature of the proposed condition. In the full-compliance case, all three upcoming IBRs maintain positive BDD margins over the scanned frequency range, and the corresponding closed-loop eigenvalues lie safely to the left of the prescribed decay rate boundary. In the moderate-noncompliance case, the active power control bandwidth of IBR33 is increased to 30 Hz from its base case value of 10 Hz. As a result, the margin of IBR33 falls below, so the certificate is lost, yet the eigenvalues still remain to the left of \(\mathrm{Re}(s)=\alpha\). This shows that failure of the BDD test does not necessarily imply violation of the decay rate requirement. In the severe-noncompliance case, where active power control bandwidth is further increased to 50 Hz, the violation of the BDD margin becomes much more severe. Consequently, the closed-loop SS mode crosses the prescribed decay rate boundary. Here, the IBR dynamic models and the closed-loop bandwidth-based control design are adopted from \cite{11082650,iravani}.

\begin{figure}[!t]
\centering
\includegraphics[width=1\linewidth]{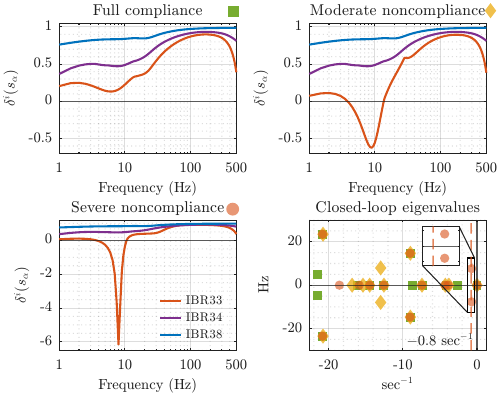}
\caption{The first three subplots show BDD margins of the upcoming IBRs for three cases: full compliance (top-left), moderate noncompliance (top-right), and severe noncompliance (bottom-left). IBR33 is made noncompliant by increasing IBR33 active power control bandwidth. The bottom-right subplot shows the corresponding closed-loop eigenvalues of the overall closed-loop system, together with the prescribed decay rate boundary $\alpha=-0.8~\textrm{sec}^{-1}$.}
\label{fig:BDD_Sufficiency_EV}
\end{figure}

Having established that the proposed condition is sufficient and therefore inherently conservative, Fig.~\ref{fig:conservativeness} next illustrates how this conservativeness depends on IBR location and how it compares with the corresponding SDD test. In Fig.~\ref{fig:conservativeness}, the left subplot shows that conservativeness depends strongly on the relative locations of the upcoming IBRs. When the second upcoming IBR is placed at bus 34, which is electrically close to bus 33, the certificate becomes restrictive. Here, conservativeness is assessed by sweeping the PLL bandwidth of the second GFL IBR while keeping the first upcoming IBR fixed at bus 33, and observing the range of PLL bandwidths for which the BDD margin remains positive. In practice, the PLL bandwidth is a key tuning parameter that directly affects synchronization dynamics and is often a trade-off between speed of response and well-damped behavior. In this sense, the admissible PLL-bandwidth range serves as a measure of design flexibility: the wider the range over which the certificate is satisfied, the less conservative the setup. For the electrically close pair formed by buses 33 and 34, the allowable PLL-bandwidth range is restricted to 61 Hz, indicating that strong coupling between the two locations rapidly consumes the available margin. In contrast, when the second upcoming IBR is placed at bus 38, which is electrically more distant, the margin remains comfortably positive over the same sweep. This is because close-by locations exhibit stronger mutual coupling, that is, larger off-diagonal transfer impedance, which makes \eqref{eq:bdd_condition_power} harder to satisfy.

% The right subplot of Fig.~\ref{fig:conservativeness} compares the proposed BDD-based condition against the corresponding SDD test. As explained by the theory in Section~\ref{subsec:BDD_theory}, the BDD margin remains positive over a wider parameter range than the SDD margin, confirming that BDD is the less conservative than SDD.
The right subplot of Fig.~\ref{fig:conservativeness} compares the outcomes of the BDD test against the SDD test. The BDD margin remains positive over a wider parameter range than the SDD margin, indicating that BDD is less conservative than SDD. This is consistent with Proposition~\ref{blockbetter}, as proved in Appendix~\ref{proofnonconservative}.
\begin{figure}[!t]
\centering
\includegraphics[width=1\linewidth]{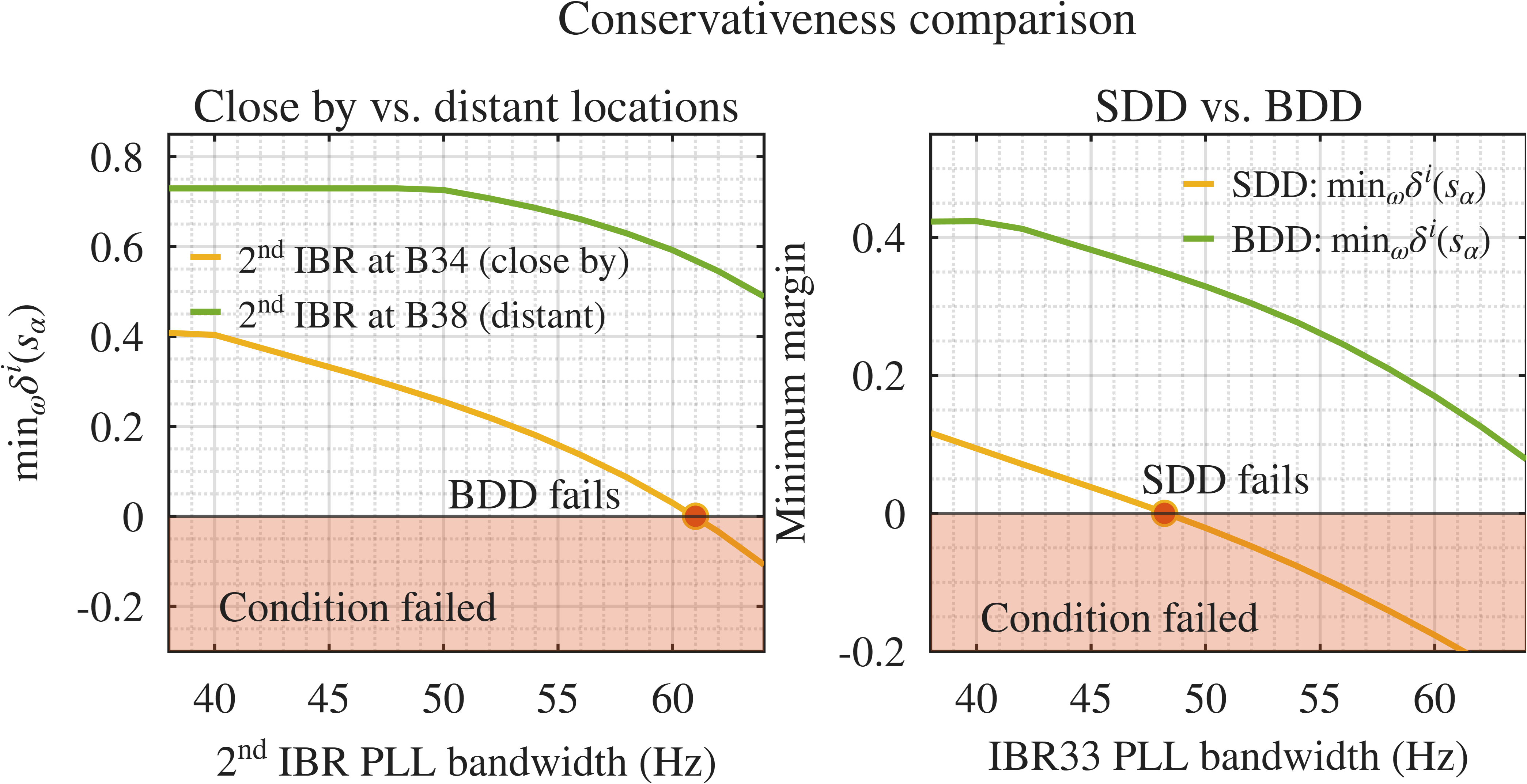}
\caption{Left: effect of location on BDD conservativeness. Right: comparison of SDD and BDD conservativeness}
\label{fig:conservativeness}
\end{figure}

%%%%%%%%%%%%%%%%%%%%%%%%%%%%%%%%%%%%%%

\subsection{Monte Carlo–Based Validation Under Diverse Operating Conditions}
To further demonstrate the generality and applicability of the proposed method for practical grid-connection compliance, the diagonal-dominance-based stability conditions are evaluated over a large set of operating points in the modified IEEE 39-bus system. Starting from the base case, the active power setpoints of IBRs and the loads are randomly perturbed according to Gaussian distributions centered at their nominal values, with standard deviations set to 20\% of the corresponding mean values. Only steady-state feasible cases are retained for subsequent analysis, that is, cases for which the power flow converges, all generator reactive power outputs remain within their limits, and all bus-voltage magnitudes stay within the prescribed admissible range of $\pm 7\%$ around nominal values.

Table~\ref{tab:dd_stats} summarizes the results for both stability and minimum decay rate certification scenarios using SDD- and BDD-based conditions. For the stability certification scenarios, all unstable cases are successfully identified and rejected by both the SDD- and BDD-based conditions, which is consistent with their nature as sufficient stability criteria. Among the total 2056 stable scenarios identified by eigenvalue analysis, the SDD-based condition certifies 378 cases, corresponding to 18.43\% of all stable scenarios. In contrast, the proposed BDD-based condition certifies 1646 stable scenarios, covering 80.06\% of the actual stable set, which is nearly five times higher than that achieved by the SDD-based condition.

For the minimum decay rate certification scenarios, $(\alpha<-0.8~\textrm{sec}^{-1})$ corresponding to a settling time of approximately 5 $s$, such that all closed-loop eigenvalues must lie to the left of $-0.8$ on the real axis. Similar to the stability certification results, both SDD- and BDD-based conditions successfully reject all poorly damped scenarios. Among the scenarios that satisfy the minimum decay rate requirement according to eigenvalue analysis, the SDD-based condition certifies 293 cases, corresponding to 14.84\%, while the BDD-based condition certifies 1547 cases, covering 78.37\% of the total well-damped set.

These results demonstrate that both SDD- and BDD-based conditions can provide stability and minimum decay rate guarantees in a fully decentralized manner. Furthermore, they highlight a significant reduction in conservativeness when using the BDD-based conditions over a wide range of operating points. This is consistent with the trend observed in Fig.~\ref{fig:conservativeness}, and with Proposition~\ref{blockbetter}, proved in Appendix~\ref{proofnonconservative}, both of which indicate that BDD is less conservative than SDD in decentralized applications.

\begin{table}[t]
\caption{Statistical performance comparison of SDD and BDD across cases under diverse operating points.}
\label{tab:dd_stats}
\centering
\setlength{\tabcolsep}{6pt}
\renewcommand{\arraystretch}{1.15}
\begin{tabular}{c|c|c|c}
\noalign{\hrule height 1.1pt}
\textbf{Case category} & \textbf{Total} & \textbf{SDD passed} & \textbf{BDD passed} \\
\hline
Stable ($\alpha < 0$)    
& 2056 
& \begin{tabular}[c]{@{}c|c@{}}379 & 18.43\%\end{tabular}
& \begin{tabular}[c]{@{}c|c@{}}\textbf{1646} & \textbf{80.06\%}\end{tabular} \\
\hline
Damped ($\alpha < -0.8$) 
& 1974 
& \begin{tabular}[c]{@{}c|c@{}}293 & 14.84\%\end{tabular}
& \begin{tabular}[c]{@{}c|c@{}}\textbf{1547} & \textbf{78.37\%}\end{tabular} \\
\noalign{\hrule height 1.1pt}
\end{tabular}
\vspace{0.5ex}

\footnotesize{Note: ``Total'' cases are determined by eigenvalue analysis; each entry under ``SDD passed'' and ``BDD passed'' reports the passed case number and corresponding percentage computed as (number of cases passed the condition / number of total cases)$\times 100\%$.}
\end{table}

% \begin{figure}[!t]
% \centering
% \includegraphics[width=1\linewidth]{New_Figures/fig_montecarlo.pdf}
% \caption{Statistical performance comparison of DD (left) and BDD (right) across 2056 operating points. All operating points are load-flow feasible and stable determined by eigenvalue analysis.}
% \label{fig:montecarlo}
% \end{figure}

% \begin{figure}[!t]
% \centering
% \includegraphics[width=1\linewidth]{New_Figures/fig_venn_montecarlo.pdf}
% \caption{Statistical performance comparison of SDD (left) and BDD (right) across 2056 operating points. All operating points are load-flow feasible and stable determined by eigenvalue analysis.}
% \label{fig:venn_montecarlo}
% \end{figure}

\section{Conclusion}
This paper presented a block diagonal dominance (BDD) framework for decentralized small-signal stability assessment of IBR-dominated power systems. The proposed criterion provides a locally verifiable IBR connection compliance condition that is theoretically grounded and practically implementable without any restrictive assumptions on network or IBR models. The results highlight that BDD enables decentralized stability certificate in cases where conventional local checks (e.g., Nyquist-based) may fail to guarantee overall system stability for multiple IBR connections. The BDD-based condition is shown to be significantly less conservative than strict diagonal dominance (SDD) with respect to admissible IBR dynamics (characterized by control bandwidths) and proximity of IBR connection points. Extensive scenario analysis on the IEEE 39-bus system shows that the BDD-based condition is satisfied for approximately 80\% of actually stable cases, compared to only 19\% for SDD. With an added minimum decay rate constraint (ensuring $\le$5 sec settling time), these reduce slightly to 78\% and 15\%, respectively, demonstrating robustness under stricter performance requirements.
Overall, the proposed BDD-based framework offers a scalable and less conservative basis for decentralized stability and performance certification of connecting IBRs.

% \section*{Acknowledgment}
% The authors would thank Prof Mark O'Malley and colleagues in \href{https://www.imperial.ac.uk/uk-electric-power-innovation/}{EPICS global centre}, and \href{https://www.neso.energy/}{NESO} (grid operator in Great Britain) for their feedback.

\ifCLASSOPTIONcaptionsoff
  \newpage
\fi

\bibliographystyle{IEEEtran}
\bibliography{DD_paper_refs}

\section*{Appendices}
\refstepcounter{appx}\label{blockthmproof}
\noindent\textbf{Appendix \theappx: Proof of Theorem \ref{blockDD}.}
By contradiction, assume $M$ is singular. There exists non-zero $v \in \mathbb{C}^n$, such that
$M v=0$. Partition $v$ according to the blocks in $M$, i.e. $v'=[v_1',v_2', \ldots, v_Q']'$.
We let $v_{-i}$ denote $[v_1',v_2',\ldots,  v_{i-1}',v_{i+1}', \ldots, v_Q' ]$.
Then $Mv=0$ implies $0 = \sum_{j=1}^Q M_{i,j} v_j = M_{i,i} v_i + M_{i,-i} v_{-i}$ for all $i
\in \{1, \ldots, Q\}$. Pick $\bar{\imath}$ so as to maximize $|v_i|_{\infty}$. Clearly $|v_{\bar{\imath}}|_{\infty} >0$, (as $v \neq 0$).
Since $M_{i,i}$ is invertible for all $i$ we get: $v_{\bar{\imath}} = - M_{\bar{\imath},\bar{\imath}}^{-1} M_{\bar{\imath},-\bar{\imath}} v_{-\bar{\imath}}$.
Moreover $|v_{-\bar{\imath}}|_{\infty} = \max_{j \neq \bar{\imath}} | v_j|_{\infty} \leq | v_{\bar{\imath}} |_{\infty}$.
The following series of inequalities applies:
\[ | v_{\bar{\imath}} |_{\infty} = | M_{\bar{\imath},\bar{\imath}}^{-1} M_{\bar{\imath},-\bar{\imath}} \, v_{-\bar{\imath}}|_{\infty} \leq \|  M_{\bar{\imath},\bar{\imath}}^{-1}  \, M_{\bar{\imath},-\bar{\imath}} \|_{\infty} \, | v_{-\bar{\imath}}|_{\infty}  \]
\[\qquad \qquad \qquad 
\leq \|  M_{\bar{\imath},\bar{\imath}}^{-1} \, M_{\bar{\imath},-\bar{\imath}} \|_{\infty}   | v_{\bar{\imath}}|_{\infty}.   
\]
%\[ \qquad \qquad \leq  \|  M_{\bar{\imath},\bar{\imath}}^{-1} \|_{\infty} \, \| M_{\bar{\imath},-\bar{\imath}} \|_{\infty}   | v_{\bar{\imath}}|_{\infty}.   \]
Dividing by $|v_{\bar{\imath}}|_{\infty}$ both sides of the previous inequality yields 
\[ 1 \leq \| M_{\bar{\imath},\bar{\imath}}^{-1} \, M_{\bar{\imath},-\bar{\imath}} \|_{\infty}, \]
which violates condition (\ref{relaxedDD}). This completes the proof.
\hfill $\Box$ 

\medskip

\refstepcounter{appx}\label{proofnonconservative}
\noindent\textbf{Appendix \theappx: Proof of Proposition \ref{blockbetter}.}
Let $A$ and $R$ be defined as follows:
\[
A= \left [ \begin{array}{cc}  m_{11} & m_{12} \\ m_{21} & m_{22} \end{array} \right ],
\quad
R = \left [ \begin{array}{cccc} m_{13} & m_{14} & \ldots& m_{1n} \\ m_{23} & m_{24} & \ldots & m_{2n} \end{array} \right ].
\]
Denote $R_1 = \sum_{j=3}^n |m_{1j}|$ and $R_2 = \sum_{j=3}^n |m_{2j}|$ as the sums of the absolute values of the external entries. The strict diagonal dominance (SDD) conditions imply:
\begin{align}
    |m_{11}| &> |m_{12}| + R_1 \implies R_1 < |m_{11}| - |m_{12}| \label{eq:sdd1} \\
    |m_{22}| &> |m_{21}| + R_2 \implies R_2 < |m_{22}| - |m_{21}| \label{eq:sdd2}
\end{align}
The inverse of the $2 \times 2$ block $A$ is given by:
\begin{equation*}
    A^{-1} = \frac{1}{\det(A)} \begin{pmatrix} m_{22} & -m_{12} \\ -m_{21} & m_{11} \end{pmatrix}
\end{equation*}
By the reverse triangle inequality, the magnitude of the determinant satisfies:
\begin{equation*}
    |\det(A)| = |m_{11}m_{22} - m_{12}m_{21}| \geq |m_{11}m_{22}| - |m_{12}m_{21}|
\end{equation*}
Note that $|m_{11}m_{22}| > |m_{12}m_{21}|$ is guaranteed by the SDD conditions. The $L_1$ norm of the first row of the product $A^{-1}R$ satisfies:
\begin{align*}
    \sum_{j=3}^n |(A^{-1}R)_{1j}|
    &= \sum_{j=3}^n \left| \frac{m_{22}m_{1j} - m_{12}m_{2j}}{\det(A)} \right| \\
    &\leq \frac{|m_{22}| \sum_{j=3}^n |m_{1j}| + |m_{12}| \sum_{j=3}^n |m_{2j}|}{|\det(A)|} \\
    &= \frac{|m_{22}|R_1 + |m_{12}|R_2}{|\det(A)|}.
\end{align*}
Substituting the bounds from \eqref{eq:sdd1} and \eqref{eq:sdd2}:
\begin{align*}
    \sum_{j=3}^n |(A^{-1}R)_{1j}|
    &< \frac{|m_{22}|(|m_{11}| - |m_{12}|) + |m_{12}|(|m_{22}| - |m_{21}|)}{|m_{11}m_{22}| - |m_{12}m_{21}|} \\
    &= \frac{|m_{11}m_{22}| - |m_{22}m_{12}| + |m_{12}m_{22}| - |m_{12}m_{21}|}{|m_{11}m_{22}| - |m_{12}m_{21}|} \\
    &= \frac{|m_{11}m_{22}| - |m_{12}m_{21}|}{|m_{11}m_{22}| - |m_{12}m_{21}|} = 1.
\end{align*}
By symmetry, the same result holds for the second row. Since both row sums are strictly less than $1$, we conclude that $\|A^{-1}R\|_{\infty} < 1$.
\hfill $\Box$

\end{document}